\documentclass[journal]{IEEEtran}

\usepackage{amsfonts}
\usepackage{amssymb}
\usepackage{amsthm}
\usepackage{amsmath,amsfonts,amssymb}
\usepackage[dvips]{graphicx}
\usepackage{subfigure}
\usepackage{verbatim}
\usepackage{setspace}
\usepackage{bm}
\usepackage{algorithmic} 
\usepackage[ruled,vlined]{algorithm2e}
\usepackage{cite}

\usepackage{changepage}
\usepackage{pdfpages}
\usepackage{color}
\newtheorem{theorem}{Theorem}

\newcommand*{\e}{\mathop{}\!\mathrm{e}}
\newcommand*{\jj}{\mathop{}\!\mathrm{j}}
\allowdisplaybreaks

\newcommand{\figwidth}{7.2}
\IEEEoverridecommandlockouts

\begin{document}
\title{\huge Multiuser Communications Aided by Cross-Linked Movable Antenna Array: Architecture and Optimization}
\author{Lipeng Zhu, ~\IEEEmembership{Member,~IEEE,}
		He Sun, ~\IEEEmembership{Member,~IEEE,}
		Wenyan Ma,~\IEEEmembership{Graduate Student Member,~IEEE,}
		Zhenyu Xiao,~\IEEEmembership{Senior Member,~IEEE,}
		and Rui Zhang,~\IEEEmembership{Fellow,~IEEE}
	\vspace{-0.5 cm}
	\thanks{L. Zhu, H. Sun, and W. Ma are with the Department of Electrical and Computer Engineering, National University of Singapore, Singapore 117583 (e-mail: zhulp@nus.edu.sg, sunele@nus.edu.sg, wenyan@u.nus.edu).}
	\thanks{Z. Xiao is with the School of Electronic and Information Engineering, Beihang University, Beijing, China 100191 (e-mail: xiaozy@buaa.edu.cn).}
	\thanks{R. Zhang is with School of Science and Engineering, Shenzhen Research Institute of Big Data, The Chinese University of Hong Kong, Shenzhen, Guangdong 518172, China (e-mail: rzhang@cuhk.edu.cn). He is also with the Department of Electrical and Computer Engineering, National University of Singapore, Singapore 117583 (e-mail: elezhang@nus.edu.sg).}
}

\maketitle


\begin{abstract}	
	Movable antenna (MA) has been regarded as a promising technology to enhance wireless communication performance by enabling flexible antenna movement. However, the hardware cost of conventional MA systems scales with the number of movable elements due to the need for independently controllable driving components. To reduce hardware cost, we propose in this paper a novel architecture named cross-linked MA (CL-MA) array, which enables the collective movement of multiple antennas in both horizontal and vertical directions. To evaluate the performance benefits of the CL-MA array, we consider an uplink multiuser communication scenario. Specifically, we aim to minimize the total transmit power while satisfying a given minimum rate requirement for each user by jointly optimizing the horizontal and vertical antenna position vectors (APVs), the receive combining at the base station (BS), and the transmit power of users. A globally lower bound on the total transmit power is derived, with closed-form solutions for the APVs obtained under the condition of a single channel path for each user. For the more general case of multiple channel paths, we develop a low-complexity algorithm based on discrete antenna position optimization. Additionally, to further reduce antenna movement overhead, a statistical channel-based antenna position optimization approach is proposed, allowing for unchanged APVs over a long time period. Simulation results demonstrate that the proposed CL-MA schemes significantly outperform conventional fixed-position antenna (FPA) systems and closely approach the theoretical lower bound on the total transmit power. Compared to the instantaneous channel-based CL-MA optimization, the statistical channel-based approach incurs a slight performance loss but achieves significantly lower movement overhead, making it an appealing solution for practical wireless systems.	
\end{abstract}

%
\IEEEpeerreviewmaketitle

\section{Introduction}
\IEEEPARstart{W}{ith} the rapid growth in communication demands with more stringent requirements, movable antenna (MA), also known as fluid antenna system (FAS), has emerged as a promising technology for next-generation wireless communication systems \cite{zhu2024historical,wong2022bruce,zhu2023MAMag}. Different from conventional fixed-position antennas (FPAs), MAs enable the flexible movement of antennas within designated transmitter or receiver regions, allowing for the full exploitation of the spatial variations in wireless channels to enhance communication performance \cite{zhu2023MAMag,zhu2022MAmodel,ma2022MAmimo}. With this reconfigurable mechanism to adapt to various signal propagation environments, MA systems can effectively increase the spectral efficiency of future wireless networks by making full use of the given number of antennas and radio frequency (RF) chains, without the need for bearing excessively high antenna and RF costs. Similar concepts of adjustable antenna positioning have also been explored in other reconfigurable antenna systems, such as fluid antennas \cite{Wong2021fluid,wong2020limit,New2024fluid,wu2024fluidMag}, flexible antennas \cite{yang2024flexible,zheng2024flexible}, and pinching antennas \cite{ding2024pinching,ouyang2025pinching}.

The advantages of MAs/FAS over conventional FPAs have been validated in various wireless systems \cite{zhu2025tutorial, shao2025tutorial}. In \cite{Wong2021fluid}, the spatial-correlation channel model was adopted in FAS to characterize the channel responses over discrete antenna ports at the receiver under rich-scattering environments. In \cite{zhu2022MAmodel}, a field-response channel model tailored for MA systems was proposed, characterizing the wireless channel as a continuous function of the MA positions at both the transmitter and receiver. The two channel models facilitate performance analysis and antenna position optimization in MA/FAS-aided wireless communication systems. The performance gains of MA/FAS-aided single-input single-output (SISO) systems have been demonstrated in both narrowband and wideband scenarios \cite{zhu2022MAmodel, Wong2021fluid,wong2020limit, zhu2024wideband}. It was shown that by reconfiguring antenna positions, the complex coefficients of multiple channel paths can be constructively combined to significantly enhance the received signal-to-noise ratio (SNR) \cite{zhu2025tutorial, shao2025tutorial}. Such an advantage in boosting received signal power can also be observed in MA-aided multiple-input single-output (MISO) systems via joint position optimization for multiple MAs \cite{mei2024movable}.

In addition to spatial diversity gains, MAs/FAS can also significantly improve spatial multiplexing performance of multiple-input multiple-output (MIMO) systems \cite{zhu2025tutorial, shao2025tutorial, ma2022MAmimo,chen2023joint,New2024MIMOFAS}. Specifically, the optimization of MAs' positions at the transmitter/receiver can effectively reshape MIMO channel matrices, the singular values of which can be flexibly balanced to improve MIMO capacity under different SNR conditions. In the context of multiuser communication, the adoption of MAs/FAS can significantly decrease the channel correlation between multiple users such that their interference can be more easily mitigated via transmit precoding or receive combining \cite{zhu2023MAmultiuser,xiao2023multiuser,wu2024globallyMA,Feng2024MAweighted,Xu2024FASmultiple}. Moreover, it was demonstrated that the two-timescale design of MA-aided multiuser communication systems can still achieve considerable performance improvements compared to FPA systems \cite{Hu20242024twotimeMA,zheng2024twotimeMA}, where the MA positions are optimized based on statistical channels to reduce antenna movement overhead. Even for pure line-of-sight (LoS) channels, an MA array can leverage the additional degrees of freedom (DoFs) in antenna position optimization to realize flexible beamforming, yielding superior performance in terms of interference nulling \cite{zhu2023MAarray,Hu2024MAarrayleak}, multi-beam forming \cite{ma2024multi,Kang2024DeepMA}, flexible beam coverage \cite{zhu2024dynamic,wang2024flexible}, etc. Recently, the six-dimensional MA (6DMA) system was proposed to fully exploit the spatial DoFs in three-dimensional (3D) position and 3D rotation/orientation \cite{shao20246DMA,shao2024discrete}. The advantages of both centralized and distributed 6DMAs in improving multiuser communication performance have been validated \cite{shao2024discrete,shao20246DMANet,shi20246DMAcellfree,pi20246DMAcoordi}. As a simplified implementation of 6DMA, rotatable antennas have also been studied to reduce hardware cost \cite{wu2024Rotatable,zheng2025rotatable}. Due to their great potentials for performance enhancement, MAs/6DMAs have been applied to various wireless systems in the literature, such as integrated sensing and communication \cite{ma2024MAsensing,lyu2025movableISAC,li2024MAISACMag,ma2025MAISAC,khalili2024advanced,Wu2024MARISISAC}, mobile edge computing \cite{ChenPC_MA_WPT_MEC,xiu2024delayMAMEC,xiu2024latencyMAMEC}, physical-layer security \cite{hu2024secure,tang2024secure,Ding2024MAsecure}, wireless power transfer \cite{Xiao2024MApower}, cognitive radio \cite{wei2024joint,zhou2024MASymbiotic}, satellite communications \cite{zhu2024dynamic,lin2024power}, unmanned aerial vehicle communications \cite{kuang2024movableISAC,liu2024uav,ren20246DMAUAV}, terahertz communications \cite{wang2024movable,zhu2024suppressing}, among others.

There are various implementation methods to realize antenna movement, such as motor-based MAs \cite{zhu2023MAMag}, micro-electromechanical system (MEMS)-based MAs \cite{balanis2008mems}, and liquid-based MAs \cite{paracha2019liquid}. In general, the motor-based MAs can achieve long-range movement, which are more suitable to be installed on infrastructures such as large-scale machines and base stations (BS), e.g., downtilt antennas. In comparison, MEMS-based MAs are featured by their fast response speed and high positioning accuracy, which are promising to be employed at small-scale devices. In addition, the element-level and array-level movement architectures were introduced in \cite{shao20246DMA,ning2024movable} to balance different trade-offs between antenna movement flexibility and hardware complexity. However, for both architectures, the movable unit requires an independent driving component for positioning reconfiguration over each dimension \cite{ning2024movable}, which results in the hardware complexity scaling with the total number of movable elements/subarrays. In particular for two-dimensional (2D) arrays on a surface, multiple movable elements/subarrays can usually conduct local movement to avoid the overlap between different driving components in practical implementation \cite{chen2023joint,zhang2024MAhybrid,ning2024movable}, which limits the DoFs in antenna position optimization. Thus, more efficient architectures of MA arrays are desired to achieve low hardware complexity while maintaining high communication performance.

\begin{figure}[t]
	\begin{center}
		\includegraphics[width=\figwidth cm]{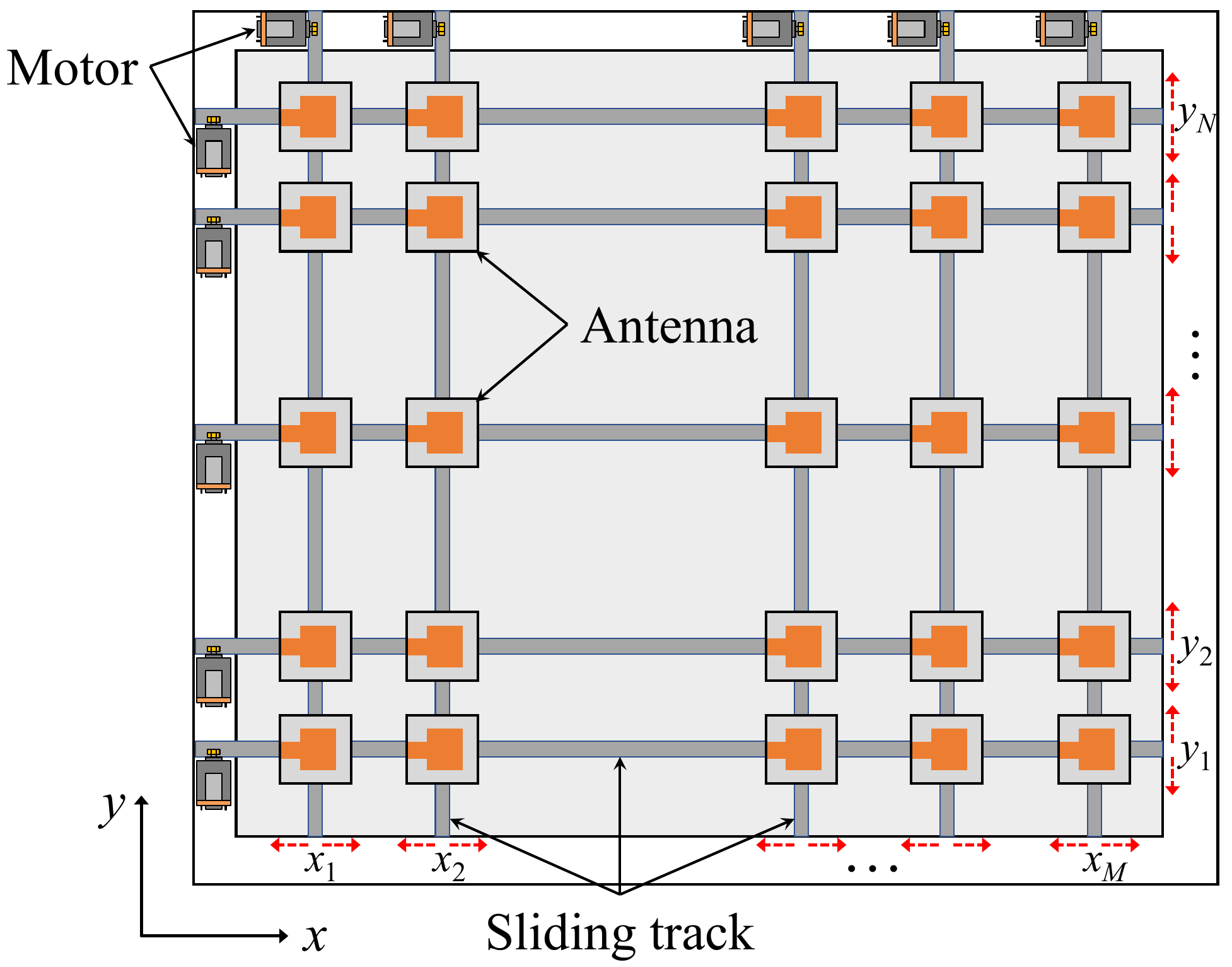}
		\caption{The architecture of the proposed CL-MA array.}
		\label{fig:architecture}
	\end{center}
\end{figure}

To address this challenge, we propose in this paper a novel architecture of cross-linked MA (CL-MA) array. As shown in Fig. \ref{fig:architecture}, the CL-MA array consists of multiple horizontal and vertical sliding tracks on the 2D plane. Each MA element/subarray is installed at the cross point of a horizontal track and a vertical track. With the aid of a motor mounted on the upper side of the array (see Fig. \ref{fig:architecture}), each vertical track can move along the horizontal direction, which changes the horizontal positions of all MAs in its corresponding column. Similarly, each row of MAs can collectively move along the vertical direction with the aid of a motor mounted on the left side of the array. Note that for conventional MA arrays with element/subarray-wise independent movement, the number of motors required for 2D movement is at least twice the total number of antennas/subarrays (e.g., an $M\times N$ 2D array with element-wise movement requires $2MN$ motors with the 2D-moving capability). In comparison, the required number of motors for the proposed CL-MA architecture is equal to the number of horizontal and vertical tracks (i.e., $M+N$ motors for the same $M\times N$ 2D array), which can significantly reduce hardware cost and control complexity for 2D arrays with large values of $M$ and $N$. 

To examine the performance advantages of the proposed CL-MA array, we investigate its application in multiuser communication systems. The main contributions of this paper are summarized as follows:
\begin{itemize}
	\item We consider an uplink communication scenario where the BS employs a CL-MA array to enhance the performance of multiple single-antenna users. The channel vector between the BS and each user is modeled as a continuous function of the horizontal and vertical antenna position vectors (APVs) of the CL-MA array. To fully exploit the spatial DoFs, the APVs are jointly optimized with the combining matrix at the BS and the users’ transmit power, aiming to minimize the total transmit power while ensuring a minimum rate requirement for each user.
	\item To characterize the performance limit of the considered system, a globally lower bound on the transmit power of each user is derived. In particular, for the case of a single channel path for each user, we reveal the condition on the numbers of users and antennas under which the lower bound is tight. Optimal solutions for the APVs for achieving this lower bound are derived in closed form. For the general case of multiple channel paths, we develop a low-complexity algorithm to conduct discrete antenna position optimization. 
	\item To further reduce antenna movement overhead, we propose a statistical channel-based antenna position optimization approach. Given the required statistical channel knowledge, the APVs are optimized and remain unchanged over a long time period to minimize the expected total transmit power of users. The Monte Carlo simulation method is employed to approximate the expectation on the total transmit power, and then the proposed algorithm of discrete antenna position optimization is applied to solve this problem.
	\item Simulation results show that the proposed CL-MA scheme significantly outperforms conventional FPA systems based on either dense or sparse arrays and closely approaches the theoretical lower bound on the total transmit power of users. Furthermore, its performance loss compared to the element-wise MA scheme is negligible. Notably, the statistical channel-based antenna position optimization achieves an appealing performance for CL-MA arrays with significantly reduced antenna movement overhead, thus offering a practically efficient solution for their implementation.
\end{itemize}

The remainder of this paper is organized as follows. Section II introduces the architecture of the CL-MA array along with the system and channel models. In Section III, we analyze the lower bound on the total transmit power of users and derive the optimal APVs for the multiuser communication system enabled by CL-MA arrays. Section IV details the proposed algorithm for discrete antenna position optimization, while Section V presents the solution for statistical channel-based antenna position optimization. Simulation results are presented and discussed in Section VI, and conclusions are finally drawn in Section VII.

\textit{Notation}: $a$, $\mathbf{a}$, $\mathbf{A}$, and $\mathcal{A}$ denote a scalar, a vector, a matrix, and a set, respectively. $(\cdot)^{*}$, $(\cdot)^{\rm{T}}$, $(\cdot)^{\rm{H}}$, and $(\cdot)^{\dagger}$ denote the conjugate, transpose, conjugate transpose, and Moore-Penrose inverse, respectively. $[\mathbf{a}]_{n}$ denotes the $n$-th entry of vector $\mathbf{a}$. $[\mathbf{A}]_{i,j}$ and $[\mathbf{A}]_{:,j}$ denote the entry in row $i$ and column $j$ and the $j$-th column vector of matrix $\mathbf{A}$, respectively. $\mathbb{Z}^{M \times N}$, $\mathbb{R}^{M \times N}$, and $\mathbb{C}^{M \times N}$ represent the sets of integer, real, and complex matrices/vectors of dimension $M \times N$, respectively. $\Re(\cdot)$, $|\cdot|$, and $\angle(\cdot)$ denote the real part, the amplitude, and the phase of a complex number or complex vector, respectively. $|\mathcal{A}|$ is the cardinality of set $\mathcal{A}$. $\varPhi$ denotes the empty set. $\mathcal{A} \cap \mathcal{B}$ and $\mathcal{A} \cup \mathcal{B}$ represent the intersection set and the union set of $\mathcal{A}$ and $\mathcal{B}$, respectively. $\|\cdot\|_{1}$ and $\|\cdot\|_{2}$ denote the 1-norm and 2-norm of a vector, respectively. $\|\cdot\|_{\mathrm{F}}$ represents the Frobenius norm of a matrix. $\otimes$ and $\odot$ denote the Kronecker product and the Khatri–Rao product (i.e., column-wise Kronecker product), respectively. $\mathrm{diag}\{\mathbf{a}\}$ represents a diagonal matrix with the diagonal entries given by $\mathbf{a}$. $\mathbf{1}_{M \times 1}$ denotes an $M$-dimensional column vector with all elements equal to 1. $\mathbf{I}_{M}$ denotes the $M$-dimensional identity matrix. $\mathbf{a} \sim \mathcal{CN}(\mathbf{0}, \mathbf{\Omega})$ indicates that $\mathbf{a}$ is a circularly symmetric complex Gaussian (CSCG) random vector with mean zero and covariance $\mathbf{\Omega}$. $\mathbb{E}\{\cdot\}$ is the expectation of a random variable.

\section{System Model}

\begin{figure}[t]
	\begin{center}
		\includegraphics[width=\figwidth cm]{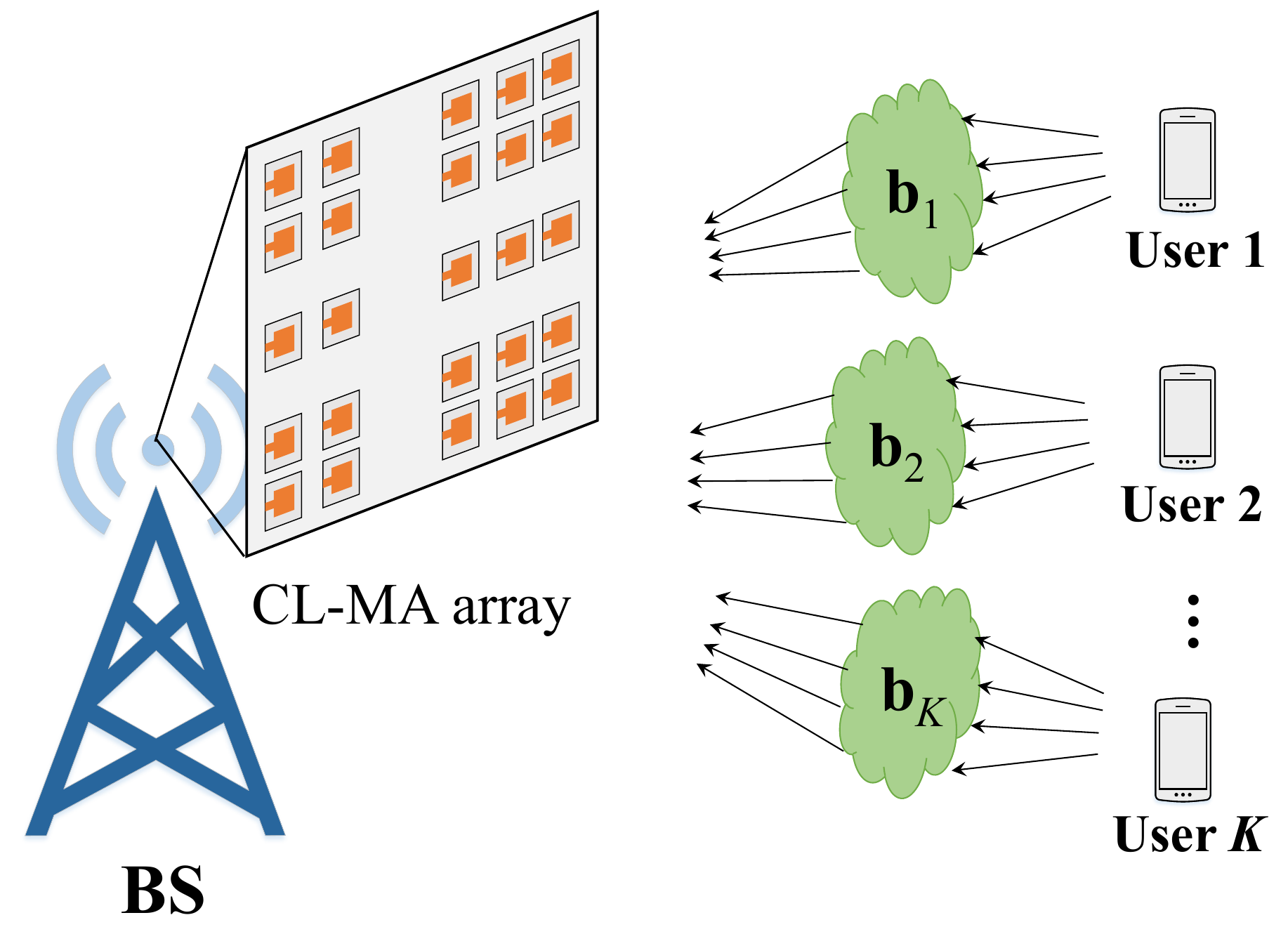}
		\caption{The proposed CL-MA array enabled BS for multiuser communication.}
		\label{fig:system}
	\end{center}
\end{figure}

We consider a sector of the BS shown in Fig. \ref{fig:system}, where a CL-MA array is employed to serve $K$ users in the uplink\footnote{Due to the duality between the uplink and downlink, the optimized positions of MAs based on uplink channels can also improve the downlink communication performance \cite{zhu2023MAmultiuser,Boche2002duality}.}. The users are each equipped with a single FPA and randomly distributed within the coverage sector of the BS. As shown in Fig. \ref{fig:architecture}, the CL-MA array consists of $M$ vertical and $N$ horizontal sliding tracks on the 2D plane. Each MA element is installed at the cross point of a horizontal track and a vertical track, which is connected to the RF chain via a flexible cable to enable movement\footnote{The proposed CL-MA architecture can also accommodate a subarray at each cross point, maintaining the same hardware complexity for the movement module while accommodating more antennas.}. Thus, the total number of antennas is $MN$. With the aid of $M$ motors mounted on the upper side of the array, the vertical tracks can move along the horizontal direction, which changes the horizontal positions of each column of MAs. Similarly, with the aid of $N$ motors mounted on the left side of the array, each row of MAs can collectively move along the vertical direction. Let $\mathbf{x}=[x_{1},x_{2},\dots,x_{M}]^{\mathrm{T}}$ and $\mathbf{y}=[y_{1},y_{2},\dots,y_{N}]^{\mathrm{T}}$ denote the horizontal and vertical APVs, respectively. Specifically, $x_{m}$ represents the horizontal coordinate of the $m$-th column of MAs, while $y_{n}$ represents the vertical coordinate of the $n$-th row of MAs. Without loss of generality, we sort the coordinates in an increasing order, i.e., $x_{1} < x_{2} < \dots <x_{M}$ and $y_{1} < y_{2} < \dots <y_{N}$. The antenna moving region is defined as a rectangle of size $[0, x_{\max}] \times [0, y_{\max}]$.

\begin{figure}[t]
	\begin{center}
		\includegraphics[width=5.0 cm]{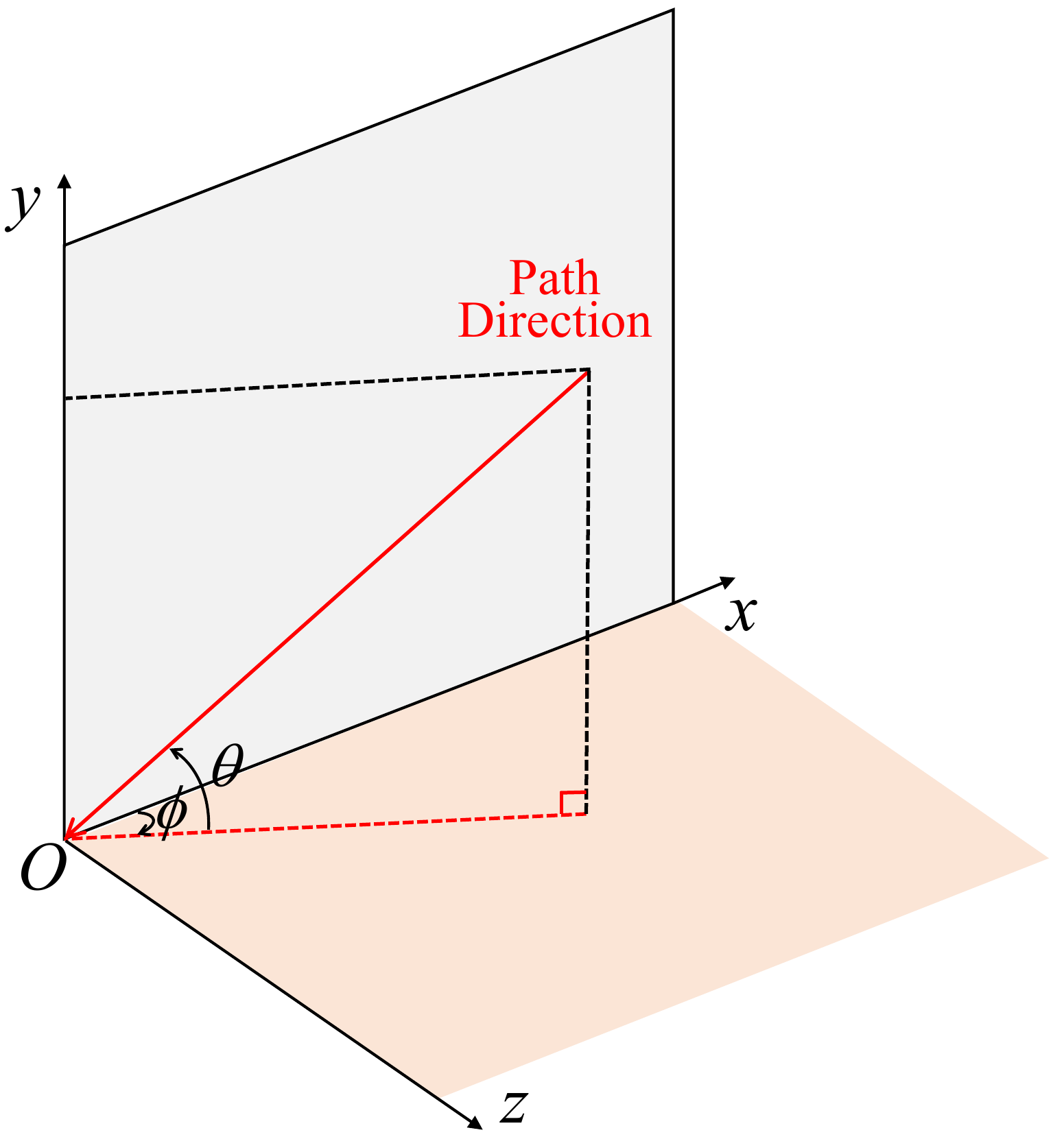}
		\caption{Illustration of the coordinate system and spatial angles at the BS.}
		\label{fig:coordinates}
	\end{center}
\end{figure}

Unlike conventional MA systems with independent positioning for each antenna element, the proposed CL-MA architecture is featured by the collective movement of antennas in each row/column. Thus, the channel vector between the BS and user $k$, $1 \leq k \leq K$, can be expressed as a function of the APVs for the CL-MA array, i.e., $\mathbf{h}_{k}(\mathbf{x}, \mathbf{y}) \in \mathbb{C}^{MN \times 1}$. For the considered uplink transmission based on spatial division multiple access (SDMA), the signal vector received at the BS is given by
\begin{equation}\label{eq_signal}
	\hat{\mathbf{s}} = \mathbf{W}^{\mathrm{H}}\mathbf{H}(\mathbf{x}, \mathbf{y})\mathbf{P}^{1/2}\mathbf{s} + \mathbf{W}^{\mathrm{H}}\mathbf{z},
\end{equation}
where $\mathbf{W} = [\mathbf{w}_{1},\mathbf{w}_{2},\dots,\mathbf{w}_{K}] \in \mathbb{C}^{MN \times K}$ is the receive combining matrix at the BS, $\mathbf{H}(\mathbf{x}, \mathbf{y}) = [\mathbf{h}_{1}(\mathbf{x}, \mathbf{y}), \mathbf{h}_{2}(\mathbf{x}, \mathbf{y}), \dots, \mathbf{h}_{K}(\mathbf{x}, \mathbf{y})] \in \mathbb{C}^{MN \times K}$ denotes the channel matrix between the $K$ users and the BS, and $\mathbf{P}^{1/2} = \mathrm{diag}\{\sqrt{p_{1}}, \sqrt{p_{2}}, \dots, \sqrt{p_{K}}\} \in \mathbb{R}^{K \times K}$ is the power scaling matrix, with $p_{k}$ denoting the transmit power of user $k$. In addition, $\mathbf{s} \sim \mathcal{CN}(0, \mathbf{I}_{K})$ and $\mathbf{z} \sim \mathcal{CN}(0, \sigma^{2}\mathbf{I}_{MN})$ are the transmit signal vector of the $K$ users with normalized power and the noise at the BS with average power $\sigma^{2}$.

\subsection{Channel Model}
We adopt the field-response channel model to characterize the variation of channel with respect to (w.r.t) antenna positions \cite{zhu2022MAmodel,ma2022MAmimo}. We assume the far-field channel condition between the BS and users since the size of the antenna moving region (in the order of several wavelengths) is usually much smaller than the signal propagation distance. For user $k$, we denote the number of channel paths received at the BS as $L_{k}$. As shown in Fig. 3, the physical elevation angle of arrival (AoA) and azimuth AoA for the $\ell$-th path, $1 \leq \ell \leq L_{k}$, of user $k$ are denoted as $\theta_{k,\ell}$ and $\phi_{k,\ell}$, respectively. Then, we define the horizontal and vertical field-response vectors (FRVs) as \cite{zhu2022MAmodel,ma2022MAmimo}
\begin{subequations}\label{eq_FRV}
	\begin{align}
		\mathbf{f}_{k}^{\mathrm{hor}}(x)&=\Big{[}\e^{\jj \frac{2\pi}{\lambda}x \vartheta_{k,1}}, \e^{\jj \frac{2\pi}{\lambda}x \vartheta_{k,2}},\dots,\e^{\jj \frac{2\pi}{\lambda}x \vartheta_{k,L_{k}}} \Big{]}^{\mathrm{T}},\\
		\mathbf{f}_{k}^{\mathrm{ver}}(y)&=\Big{[}\e^{\jj \frac{2\pi}{\lambda}y \varphi_{k,1}}, \e^{\jj \frac{2\pi}{\lambda}y \varphi_{k,2}},\dots,\e^{\jj \frac{2\pi}{\lambda}y \varphi_{k,L_{k}}} \Big{]}^{\mathrm{T}},
	\end{align}
\end{subequations}
where $\vartheta_{k,\ell} \triangleq \cos \theta_{k,\ell} \cos \phi_{k,\ell}$ and $\varphi_{k,\ell} \triangleq \sin \theta_{k,\ell} $ denote the virtual AoAs for the $\ell$-th channel path of user $k$. In particular, term $\frac{2\pi}{\lambda}x \vartheta_{k,\ell}$ represents the phase difference in the channel coefficient for the $\ell$-th path of user $k$ between horizontal antenna position $x$ and the reference point. Similarly, term $\frac{2\pi}{\lambda}y \varphi_{k,\ell}$ represents the phase difference in the channel coefficient for the $\ell$-th path of user $k$ between vertical antenna position $y$ and the reference point. 

Collecting the FRVs for all antennas' positions, we obtain the horizontal and vertical field-response matrices (FRMs) as
\begin{subequations}\label{eq_FRM}
	\begin{align}
		\mathbf{F}_{k}^{\mathrm{hor}}(\mathbf{x})&=\big{[}\mathbf{f}_{k}^{\mathrm{hor}}(x_{1}), \mathbf{f}_{k}^{\mathrm{hor}}(x_{2}), \dots, \mathbf{f}_{k}^{\mathrm{hor}}(x_{M})\big{]} \in \mathbb{C}^{L_{k} \times M},\\
		\mathbf{F}_{k}^{\mathrm{ver}}(\mathbf{y})&=\big{[}\mathbf{f}_{k}^{\mathrm{ver}}(y_{1}), \mathbf{f}_{k}^{\mathrm{ver}}(y_{2}), \dots, \mathbf{f}_{k}^{\mathrm{ver}}(y_{N})\big{]} \in \mathbb{C}^{L_{k} \times N}.
	\end{align}
\end{subequations}
Denote $\mathbf{b}_{k} \in \mathbb{C}^{L_{k} \times 1}$ as the path response vector, the $\ell$-th element of which represents the complex coefficient of the the $\ell$-th channel path from user $k$ to the reference point of the BS. Thus, the instantaneous channel vector for user $k$, $1 \leq k \leq K$, can be obtained as
\begin{equation}\label{eq_channel}
	\mathbf{h}_{k}(\mathbf{x}, \mathbf{y}) = \left(\mathbf{F}_{k}^{\mathrm{hor}}(\mathbf{x})^{\mathrm{H}} \odot \mathbf{F}_{k}^{\mathrm{ver}}(\mathbf{y})^{\mathrm{H}} \right) \mathbf{b}_{k} \in \mathbb{C}^{MN \times 1}.
\end{equation}
Thus, we have the channel matrix between the BS and users as $\mathbf{H}(\mathbf{x}, \mathbf{y}) = [\mathbf{h}_{1}(\mathbf{x}, \mathbf{y}), \mathbf{h}_{2}(\mathbf{x}, \mathbf{y}), \dots, \mathbf{h}_{K}(\mathbf{x}, \mathbf{y})]$, with $\mathbf{H}$ representing the channel mapping from the APVs to the channel responses of all users\footnote{In the sequel of this paper, we assume that the channel matrix between the BS and users is of full column rank, which holds with probability one for multiuser communication systems with randomly fading channels \cite{TseFundaWC}.}.

\subsection{Problem Formulation}
According to \eqref{eq_signal}, the received signal-to-interference-plus-noise ratio (SINR) for user $k$ is given by
\begin{equation}\label{eq_SINR}
	\gamma_{k} = \frac{|\mathbf{w}_{k}^{\mathrm{H}}\mathbf{h}_{k}(\mathbf{x}, \mathbf{y})|^{2}p_{k}}{\sum \limits _{q=1,q \neq k}^{K} |\mathbf{w}_{k}^{\mathrm{H}}\mathbf{h}_{q}(\mathbf{x}, \mathbf{y})|^{2}p_{q} + \|\mathbf{w}_{k}\|_{2}^{2}\sigma^{2}}.
\end{equation}
In this paper, we aim to minimize the total transmit power of users by jointly optimizing the APVs and receive combining matrix at the BS and the transmit power of users, with given minimum rate constraint on each user. The optimization problem is thus formulated as
\begin{subequations}\label{eq_problem_ist}
	\begin{align}
		\mathop{\min} \limits_{\mathbf{x}, \mathbf{y}, \mathbf{W}, \mathbf{p}} ~ &\sum \limits_{k=1}^{K} p_{k} \label{eq_problem_ist_a}\\
		\mathrm{s.t.}~~~~  &\log_{2}(1+\gamma_{k}) \geq r_{k},~1 \leq k \leq K,\label{eq_problem_ist_b}\\
		&p_{k} \geq 0,~1 \leq k \leq K, \label{eq_problem_ist_c}\\
		&0 \leq x_{m} \leq x_{\max},~1 \leq m \leq M, \label{eq_problem_ist_d}\\
		&x_{m+1} - x_{m} \geq d_{\min,x},~1 \leq m \leq M-1, \label{eq_problem_ist_e}\\
		&0 \leq y_{n} \leq y_{\max},~1 \leq n \leq N, \label{eq_problem_ist_f}\\
		&y_{n+1} - y_{n} \geq d_{\min,y},~1 \leq n \leq N-1, \label{eq_problem_ist_g}
	\end{align}
\end{subequations}
where constraint \eqref{eq_problem_ist_b} guarantees each user satisfying the minimum rate requirement, $r_{k}$, constraint \eqref{eq_problem_ist_c} ensures the non-negative transmit power of each user, constraints \eqref{eq_problem_ist_d} and \eqref{eq_problem_ist_f} confine the antenna moving within region $[0, x_{\max}] \times [0, y_{\max}]$, and \eqref{eq_problem_ist_e} and \eqref{eq_problem_ist_g} constrain the inter-antenna spacing over the horizontal and vertical directions being no less than minimum thresholds $d_{\min,x}$ and $d_{\min,y}$, respectively. Problem \eqref{eq_problem_ist} is challenging to be optimally solved because the achievable rate in \eqref{eq_problem_ist_b} has a non-concave form w.r.t. the highly coupled optimization variables. In the following, we first present analytical results to characterize the performance bound of the considered system, and then develop practical algorithms to obtain suboptimal solutions.

\section{Performance Analysis}
In this section, the globally lower bound on the transmit power of users in problem \eqref{eq_problem_ist} is analyzed. Then, optimal solutions for the APVs are derived in closed form under specific conditions on the number of users and channel paths. In particular, the constraints on antenna moving region, i.e., \eqref{eq_problem_ist_d} and \eqref{eq_problem_ist_f}, are omitted to facilitate performance analysis, which indicates that the antenna moving region can be arbitrarily large.

\begin{theorem}\label{theo_bound}
	A lower bound on the transmit power of each user in problem \eqref{eq_problem_ist} is given by
	\begin{equation}\label{eq_bound_power}
		\bar{p}_{k} = \frac{\sigma^{2}(2^{r_{k}-1})}{MN\|\mathbf{b}_{k}\|_{1}^{2}},~1 \leq k \leq K.
	\end{equation}
\end{theorem}
\begin{proof}
	See Appendix \ref{App_bound}.
\end{proof}

According to the proof of Theorem \ref{theo_bound}, achieving the lower bound on the transmit power in \eqref{eq_bound_power} requires two conditions on the channel vectors between the BS and users. On one hand, to guarantee the equalities in \eqref{eq_bound_rate} and \eqref{eq_bound_MRC}, the channel vectors for multiple users should be orthogonal to each other, i.e., $\mathbf{h}_{k}(\mathbf{x}, \mathbf{y})^{\mathrm{H}}\mathbf{h}_{q}(\mathbf{x}, \mathbf{y})=0$, $\forall k \neq q$, which is termed as the \emph{channel vector orthogonality (CVO) condition}. On the other hand, the channel power gain for each user should achieve its upper bound, i.e., $\|\mathbf{h}_{k}(\mathbf{x}, \mathbf{y})\|_{2}^{2}= MN\|\mathbf{b}_{k}\|_{1}^{2}$ holds in \eqref{eq_bound_channel_power}, which is termed as the \emph{maximal channel power (MCP) condition}. In addition, two implicit conditions are required, i.e., the MRC is adopted at the BS and the achievable rate for each user is equal to its minimum requirement, which can be easily achieved by setting the combining matrix as $\mathbf{W}=\mathbf{H}(\mathbf{x}, \mathbf{y})$ and the transmit power as $p_{k} = \bar{p}_{k}$, $1 \leq k \leq K$.

It has been shown in \cite{zhu2022MAmodel} that for MA-enabled SISO or MISO systems, the MCP condition can be achieved if the number of channel paths is no larger than $3$ and each antenna can be independently moved in the 2D plane. However, for the considered multiuser MIMO systems with arbitrary number of channel paths as well as the proposed CL-MA array with collective antenna movement, it is generally difficult to guarantee the MCP condition. Moreover, the CVO condition is also difficult to satisfy due to the intricate superposition of multiple channel paths shown in \eqref{eq_channel}. To facilitate performance analysis, we consider the special case with a single  channel path received at the BS for each user, i.e., $L_{k}=1$, $1 \leq k \leq K$. This scenario makes sense for communication systems operating at high-frequency bands, e.g., millimeter-wave and terahertz bands, with the wireless channels dominated by an LoS or reflected non-LoS (NLoS) path. 

Under the condition of a single (LoS/NLoS) channel path for each user, the channel vector for user $k$ in \eqref{eq_channel} is recast into 
\begin{equation}\label{eq_channel_LoS}
	\hat{\mathbf{h}}_{k}(\mathbf{x}, \mathbf{y}) = b_{k} \mathbf{a}_{k}^{\mathrm{hor}}(\mathbf{x}) \otimes \mathbf{a}_{k}^{\mathrm{ver}}(\mathbf{y}),~1 \leq k \leq K,
\end{equation}
where $b_{k}$ is the complex coefficient for the channel path. $\mathbf{a}_{k}^{\mathrm{hor}}(\mathbf{x})$ and $\mathbf{a}_{k}^{\mathrm{ver}}(\mathbf{y})$ are horizontal and vertical steering vectors given by
\begin{subequations}\label{eq_steer}
	\begin{align}
		\mathbf{a}_{k}^{\mathrm{hor}}(\mathbf{x})&=\big{[}\e^{-\jj\frac{2\pi}{\lambda}x_{1} \vartheta_{k}}, \e^{-\jj\frac{2\pi}{\lambda}x_{2} \vartheta_{k}},\dots,\e^{-\jj\frac{2\pi}{\lambda}x_{M} \vartheta_{k}} \big{]}^{\mathrm{T}},\\
		\mathbf{a}_{k}^{\mathrm{ver}}(\mathbf{y})&=\big{[}\e^{-\jj\frac{2\pi}{\lambda}y_{1} \varphi_{k}}, \e^{-\jj\frac{2\pi}{\lambda}y_{2} \varphi_{k}},\dots,\e^{-\jj\frac{2\pi}{\lambda}y_{N} \varphi_{k}} \big{]}^{\mathrm{T}},
	\end{align}
\end{subequations}
where $\vartheta_{k}$ and $\varphi_{k}$ denote the virtual AoAs for the channel path of user $k$\footnote{For multiuser communication systems with random distributions of users and scatterers, the probability that two users have exactly the same AoA is zero \cite{zhu2024wideband}. Thus, we assume that the AoAs are distinguished for different users, which can guarantee the multiuser channel matrix to have full column rank.}. In fact, the horizontal/vertical steering vector in \eqref{eq_steer} is identical to the conjugate transpose of the FRM with $L_{k}=1$. 

\begin{theorem}\label{theo_bound_condition}
	Under the condition of a single channel path for each user, the lower bound on the transmit power in \eqref{eq_bound_power} is tight for $K(K-1)/2 \leq I_{M}+I_{N}$, with $I_{M}$ and $I_{N}$ denoting the total number of prime factors of $M$ and $N$, respectively.
\end{theorem}
\begin{proof}
	See Appendix \ref{App_bound_condition}.
\end{proof}

Theorem \ref{theo_bound_condition} indicates that the lower bound on the transmit power of each user given by Theorem \ref{theo_bound} is achievable if the number of users is sufficiently small. Next, we provide optimal solutions for the APVs, i.e., $\mathbf{x}^{\star}$ and $\mathbf{y}^{\star}$, to achieve this lower bound. As can be observed from the proof of Theorem \ref{theo_bound_condition}, the optimal APVs are not unique, while they change with the different partitions between $\mathcal{P}_{1}$ and $\mathcal{P}_{2}$. For ease of exposition, we denote the $i$-th element in $\mathcal{P}_{1}$ as $(k_{1}^{(i)},q_{1}^{(i)})$, $1 \leq i \leq |\mathcal{P}_{1}| \leq I_{M}$. Similarly, the $i$-th element in $\mathcal{P}_{2}$ is denoted as $(k_{2}^{(i)},q_{2}^{(i)})$, $1 \leq i \leq |\mathcal{P}_{2}| \leq I_{N}$.

Following the method in \cite{zhu2023MAarray}, we denote the prime factorization of integer $M$ as $M=\prod_{i=1}^{I_{M}} m_{i}$, with $m_{i}$ being the $i$-th prime factor (sorted in a non-decreasing order) and $I_{M}$ representing the total number of prime factors. Then, we define $M_{1}=1$, $M_{i} = \prod_{j=1}^{i-1} m_{j}$, $2 \leq i \leq I_{M}$, and $\mathbf{m}=[M_{1}, M_{2}, \dots, M_{I_{M}}]^{\mathrm{T}}$. It is easy to verify that any positive integer $m$, $1 \leq m \leq M$, can be uniquely determined by the factorization coefficient vector $\mathbf{u}_{m} \in \mathbb{Z}^{I_{M} \times 1}$ as $m=\mathbf{u}_{m}^{\mathrm{T}}\mathbf{m}+1$, subject to $[\mathbf{u}_{m}]_{i}<m_{i}$, $1 \leq i \leq I_{M}$. Specifically, $\mathbf{u}_{m}$ is given by the integer quotients of successively dividing the remainder of number $(m-1)$ by each element in $\mathbf{m}$ (from back to front). For example, for $M=24=2\times 2 \times 2 \times 3$, we have $I_{M}=4$ and $\mathbf{m}=[1, 2, 4, 8]^{\mathrm{T}}$. Then, we can obtain $\mathbf{u}_{8}=[1, 1, 1, 0]^{\mathrm{T}}$, $\mathbf{u}_{23}=[0, 1, 1, 2]^{\mathrm{T}}$, and so on. Similarly, we denote the prime factorization of integer $N$ as $N=\prod_{i=1}^{I_{N}} n_{i}$, with $n_{i}$ being the $i$-th prime factor (sorted in a non-decreasing order) and $I_{N}$ representing the total number of prime factors. Then, we define $N_{1}=1$, $N_{i} = \prod_{j=1}^{i-1} n_{j}$, $2 \leq i \leq I_{N}$, and $\mathbf{n}=[N_{1}, N_{2}, \dots, N_{I_{N}}]^{\mathrm{T}}$. Any positive integer $n$, $1 \leq n \leq N$, can be uniquely determined by the factorization vector $\mathbf{v}_{n} \in \mathbb{Z}^{I_{N} \times 1}$ as $n=\mathbf{v}_{n}^{\mathrm{T}}\mathbf{n}+1$, subject to $[\mathbf{u}_{n}]_{i}<n_{i}$, $1 \leq i \leq I_{N}$, where $\mathbf{v}_{n}$ is given by the integer quotients of successively dividing the remainder of number $(n-1)$ by each element in $\mathbf{n}$ (from back to front). Moreover, we define $\mathbf{U}=[\mathbf{u}_{1}, \mathbf{u}_{2},\dots,\mathbf{u}_{M}] \in \mathbb{Z}^{I_{M} \times M}$ and $\mathbf{V}=[\mathbf{v}_{1}, \mathbf{v}_{2},\dots,\mathbf{v}_{N}] \in \mathbb{Z}^{I_{N} \times N}$.
\begin{theorem}\label{theo_APV}
	Under the condition of a single channel path for each user and $K(K-1)/2 \leq I_{M}+I_{N}$, optimal APVs for problem \eqref{eq_problem_ist} are given by
	\begin{equation}\label{eq_APV}
		\mathbf{x}^{\star}=\mathbf{U}^{\mathrm{T}}\mathbf{d}_{x},~\mathbf{y}^{\star}=\mathbf{V}^{\mathrm{T}}\mathbf{d}_{y},
	\end{equation}
	with 
	\begin{equation}\label{eq_dis_x}
		[\mathbf{d}_{x}]_{i}=\left\{
		\begin{aligned}
			&\frac{(\rho_{i}+1/m_{i})\lambda}{|\vartheta_{k_{1}^{(i)}}-\vartheta_{q_{1}^{(i)}}|},~1 \leq i \leq |\mathcal{P}_{1}|,\\
			&\sum_{j=1}^{i-1}(m_{j}-1)[\mathbf{d}_{x}]_{j}+d_{\min,x},~|\mathcal{P}_{1}| + 1 \leq i \leq I_{M},
		\end{aligned}
		\right.
	\end{equation}
	\begin{equation}\label{eq_dis_y}
		[\mathbf{d}_{y}]_{i}=\left\{
		\begin{aligned}
			&\frac{(\tau_{i}+1/n_{i})\lambda}{|\varphi_{k_{2}^{(i)}}-\varphi_{q_{2}^{(i)}}|},~1 \leq i \leq |\mathcal{P}_{2}|,\\
			&\sum_{j=1}^{i-1}(n_{j}-1)[\mathbf{d}_{y}]_{j}+d_{\min,y},~|\mathcal{P}_{2}| + 1 \leq i \leq I_{N},
		\end{aligned}
		\right.
	\end{equation}
	where $\rho_{i}$ is the minimum nonnegative integer ensuring $[\mathbf{d}_{x}]_{i} \geq \sum_{j=1}^{i-1}(m_{j}-1)[\mathbf{d}_{x}]_{j}+d_{\min,x}$, $1 \leq i \leq |\mathcal{P}_{1}|$, and $\tau_{i}$ is the minimum nonnegative integer ensuring $[\mathbf{d}_{y}]_{i} \geq \sum_{j=1}^{i-1}(n_{j}-1)[\mathbf{d}_{y}]_{j}+d_{\min,y}$, $1 \leq i \leq |\mathcal{P}_{2}|$.
\end{theorem}
\begin{proof}
	See Appendix \ref{App_APV}.
\end{proof}

\begin{figure}[t]
	\begin{center}
		\includegraphics[width=5.5 cm]{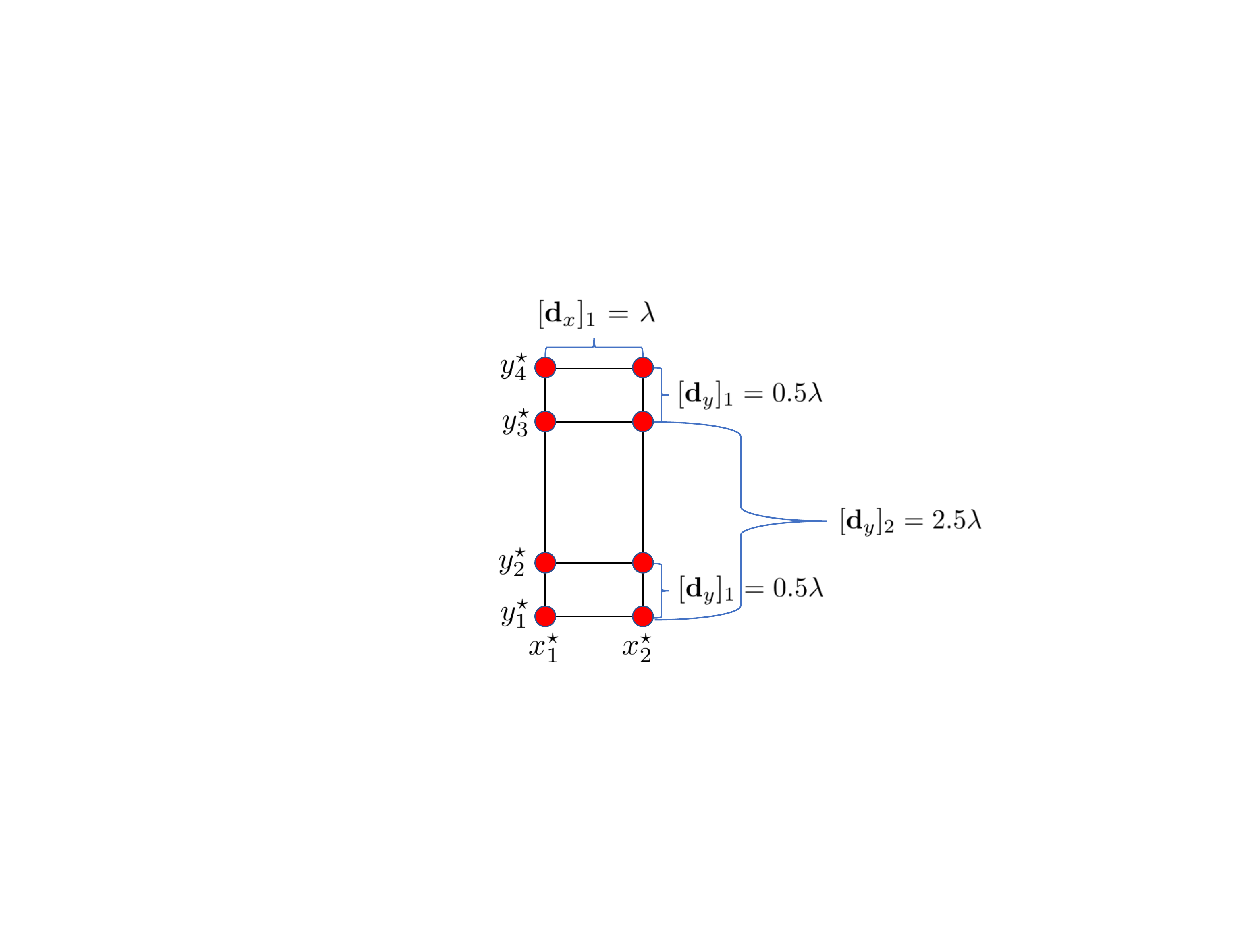}
		\caption{Illustration of the optimal APVs for $M \times N =2 \times 4$ and $K=3$.}
		\label{fig:optimalAPV}
	\end{center}
\end{figure}

The key idea for constructing the optimal APVs in Theorem \ref{theo_APV} is by ensuring the SVO condition for each pair of users sequentially, subject to the constraint on the minimum inter-antenna spacing over both horizontal and vertical directions. In Fig. \ref{fig:optimalAPV}, we show an example of the system with $M \times N =2 \times 4$ CL-MAs at the BS serving $K=3$ users, with their virtual AoAs given by $\vartheta_{1}=0.1,\varphi_{1}=-0.3$, $\vartheta_{2}=-0.4,\varphi_{2}=0.5$, $\vartheta_{3}=-0.2,\varphi_{3}=0.7$. The optimal APVs can be constructed as follows. Specifically, we set $\mathcal{P}_{1}=\{(1,2)\}$ and $\mathcal{P}_{2}=\{(1,3), (2,3)\}$. For the horizontal APV, the 1st column of MAs are deployed at $x_{1}^{\star}=0$. Then, the 2nd column of MAs are deployed at $x_{2}^{\star}=x_{1}^{\star}+[\mathbf{d}_{x}]_{1}=\lambda$ to guarantee the CVO condition for users 1 and 2. For the vertical APV, the 1st row of MAs are deployed at $y_{1}^{\star}=0$. Then, the 2nd row of MAs are deployed at $y_{2}^{\star}=y_{1}^{\star}+[\mathbf{d}_{y}]_{1}=0.5\lambda$ to guarantee the CVO condition for users 1 and 3. Finally, the 3rd and 4th rows of MAs are respectively deployed at $y_{3}^{\star}=y_{1}^{\star}+[\mathbf{d}_{y}]_{2}=2.5\lambda$ and $y_{4}^{\star}=y_{2}^{\star}+[\mathbf{d}_{y}]_{2}=3\lambda$ to guarantee the CVO for users 2 and 3. The CVO between any pair of two users is guaranteed by the constructed APVs, which are thus optimal to problem \eqref{eq_problem_ist} for achieving the lower bound on the total transmit power. As can be observed in Fig. \ref{fig:optimalAPV}, the size of the antenna moving region required to achieve the lower bound on the transmit power is only $\lambda \times 3\lambda$, demonstrating the ease of realization of the proposed optimal solution presented in Theorem \ref{theo_APV}.

Despite the lower bound on the total transmit power derived in \eqref{eq_bound_power} and the optimal solution given by Theorem \ref{theo_APV}, the results rely on the condition of a single channel path for each user. For the general scenario with multiple channel paths of users, it is challenging to derive the optimal APVs. In addition, if the number of users is large or the size of the antenna moving region is limited, the DoF in antenna position optimization may not suffice to achieve the performance bound in Theorem \ref{theo_bound}. Thus, a general algorithm is required to solve problem \eqref{eq_problem_ist} efficiently, which will be presented in the next section.

\section{Optimization Algorithm}
The main challenge for solving problem \eqref{eq_problem_ist} lies in that the optimization variables are highly coupled in the non-convex constraint \eqref{eq_problem_ist_b}. To address this problem, we start with the design of the combining matrix and the transmit power vector for any given APVs $\mathbf{x}$ and $\mathbf{y}$. Then, the optimization of APVs based on discrete approximation is presented.

\subsection{Optimization of Transmit Power and Receive Combining}
Inspired by the analytical results in Section III, the antenna position optimization can efficiently decrease the correlation of channel vectors for multiple users. Thus, the zero-forcing (ZF)-based combining can approach the performance of the optimal minimum mean square error (MMSE) receiver because the interference nulling does not result in significant loss of the desired signal power \cite{zhu2023MAarray}. For any given channel matrix $\mathbf{H}(\mathbf{x}, \mathbf{y})$, the ZF-based receive combining matrix is adopted as a suboptimal solution given by
\begin{equation}\label{eq_ZF}
	\begin{aligned}
		\mathbf{W}(\mathbf{x}, \mathbf{y}) =& \left(\mathbf{H}(\mathbf{x}, \mathbf{y})^{\mathrm{H}}\right)^{\dagger} 
		= \mathbf{H}(\mathbf{x}, \mathbf{y}) \left(\mathbf{H}(\mathbf{x}, \mathbf{y})^{\mathrm{H}}\mathbf{H}(\mathbf{x}, \mathbf{y})\right)^{-1},
	\end{aligned}
\end{equation}
which is expressed as a function of the APVs. Substituting \eqref{eq_ZF} into \eqref{eq_SINR}, the received SNR for user $k$ is obtained as
\begin{equation}\label{eq_SNR}
	\gamma_{k} = \frac{p_{k}}{\left\|\mathbf{H}(\mathbf{x}, \mathbf{y}) \left[\left(\mathbf{H}(\mathbf{x}, \mathbf{y})^{\mathrm{H}}\mathbf{H}(\mathbf{x}, \mathbf{y})\right)^{-1}\right]_{:,k}\right\|_{2}^{2}\sigma^{2}}.
\end{equation}
Then, the minimum transmit power of user $k$ can be obtained by substituting \eqref{eq_SNR} into \eqref{eq_problem_ist_b} as
\begin{equation}\label{eq_power_min}
	\begin{aligned}
		\tilde{p}_{k} = \left\|\mathbf{H}(\mathbf{x}, \mathbf{y}) \left[\left(\mathbf{H}(\mathbf{x}, \mathbf{y})^{\mathrm{H}}\mathbf{H}(\mathbf{x}, \mathbf{y})\right)^{-1}\right]_{:,k}\right\|_{2}^{2} \sigma^{2}(2^{r_{k}-1}).
	\end{aligned}
\end{equation}
Denoting $\mathbf{\Omega} = \mathrm{diag}\{\sigma^{2}(2^{r_{1}-1}), \dots, \sigma^{2}(2^{r_{K}-1})\}$, the minimum total transmit power of all the $K$ users under the ZF combining is thus given by
\begin{equation}\label{eq_total_power_min}{\small
	\begin{aligned}
		&\sum \limits_{k=1}^{K} \tilde{p}_{k} 
		= \sum \limits_{k=1}^{K} \left\|\mathbf{H}(\mathbf{x}, \mathbf{y}) \left[\left(\mathbf{H}(\mathbf{x}, \mathbf{y})^{\mathrm{H}}\mathbf{H}(\mathbf{x}, \mathbf{y})\right)^{-1}\right]_{:,k}\right\|_{2}^{2} [\mathbf{\Omega}]_{k,k}\\
		= &\left\|\mathbf{H}(\mathbf{x}, \mathbf{y}) \left(\mathbf{H}(\mathbf{x}, \mathbf{y})^{\mathrm{H}}\mathbf{H}(\mathbf{x}, \mathbf{y})\right)^{-1}\mathbf{\Omega}^{\frac{1}{2}}\right\|_{\mathrm{F}}^{2} \\
		= & \mathrm{tr} \left\{ \left(\mathbf{H}(\mathbf{x}, \mathbf{y})^{\mathrm{H}}\mathbf{H}(\mathbf{x}, \mathbf{y})\right)^{-1}\mathbf{\Omega} \right\} \triangleq \tilde{P}(\mathbf{x}, \mathbf{y}),
	\end{aligned}}
\end{equation}
which is determined by the APVs $\mathbf{x}$ and $\mathbf{y}$ only.

\subsection{Optimization of APVs}
With the closed-form solution for the receive combining and transmit power in \eqref{eq_ZF} and \eqref{eq_total_power_min}, the antenna position optimization problem can be simplified as
\begin{equation}\label{eq_problem_APV}
	\begin{aligned}
		\mathop{\min} \limits_{\mathbf{x}, \mathbf{y}} ~ &\tilde{P}(\mathbf{x}, \mathbf{y})\\
		\mathrm{s.t.}~~  &\eqref{eq_problem_ist_d},~\eqref{eq_problem_ist_e},~\eqref{eq_problem_ist_f}, ~\eqref{eq_problem_ist_g}.~
	\end{aligned}
\end{equation}

It is challenging to solve problem \eqref{eq_problem_APV} optimally because of the highly nonlinear form of the channel vector for each user w.r.t. the APVs shown in \eqref{eq_channel}. To address this issue, we adopt the discrete position optimization for the considered CL-MA systems. The key idea is to discretize the entire antenna moving region into multiple candidate positions, while each row/column of MAs should be selected from these candidate positions. If the discretization accuracy is sufficiently high, the performance of discrete position optimization can closely approach that of the continuous position optimization \cite{mei2024movable,wu2024globallyMA}. In addition, it is worth noting that the approaches of discrete position optimization in \cite{mei2024movable} and \cite{wu2024globallyMA} are tailored for MA systems with independent movement for each antenna element, which are not applicable to our considered CL-MA system with collective movement for each row/column of antennas. Thus, leveraging the unique structure of the CL-MA array, we develop a low-complexity algorithm for discrete antenna position optimization.
To this end, the horizontal region $[0, x_{\max}]$ and the vertical region $[0, y_{\max}]$ are uniformly discretized into $\bar{M}$ and $\bar{N}$ positions, respectively, which are given by
\begin{equation}\label{eq_pos_candi}
	\begin{aligned}
		&\bar{\mathbf{x}}=\left[\bar{x}_{1}, \bar{x}_{2}, \dots, \bar{x}_{\bar{M}}\right]^{\mathrm{T}},~
		\bar{\mathbf{y}}=\left[\bar{y}_{1}, \bar{y}_{2}, \dots, \bar{y}_{\bar{N}}\right]^{\mathrm{T}},
	\end{aligned}
\end{equation}
with $\bar{x}_{m}=\frac{(m-1/2)x_{\max}}{\bar{M}}$, $1 \leq m \leq \bar{M}$, and $\bar{y}_{n}=\frac{(n-1/2)y_{\max}}{\bar{N}}$, $1 \leq n \leq \bar{N}$. Then, we can obtain the channel vector of all candidate positions in the entire antenna moving region at the BS from user $k$ as
\begin{equation}\label{eq_channel_candi}
	\bar{\mathbf{h}}_{k} = (\bar{\mathbf{F}}_{k}^{\mathrm{hor}})^{\mathrm{H}} \odot (\bar{\mathbf{F}}_{k}^{\mathrm{ver}})^{\mathrm{H}} \mathbf{b}_{k} \in \mathbb{C}^{\bar{M}\bar{N} \times 1},~1 \leq k \leq K,
\end{equation}
with $\bar{\mathbf{F}}_{k}^{\mathrm{hor}} = [\mathbf{f}_{k}^{\mathrm{hor}}(\bar{x}_{1}), \mathbf{f}_{k}^{\mathrm{hor}}(\bar{x}_{2}), \dots, \mathbf{f}_{k}^{\mathrm{hor}}(\bar{x}_{\bar{M}})] \in \mathbb{C}^{L_{k} \times \bar{M}}$ and $\bar{\mathbf{F}}_{k}^{\mathrm{ver}} = [\mathbf{f}_{k}^{\mathrm{ver}}(\bar{y}_{1}), \mathbf{f}_{k}^{\mathrm{ver}}(\bar{y}_{2}), \dots, \mathbf{f}_{k}^{\mathrm{ver}}(\bar{y}_{\bar{N}})] \in \mathbb{C}^{L_{k} \times \bar{N}}$ representing the horizontal and vertical FRMs for all candidate positions, respectively. Note that $\bar{\mathbf{h}}_{k}$ is a constant vector for any given discrete positions. For convenience, we denote $\bar{\mathbf{H}}=[\bar{\mathbf{h}}_{1}, \bar{\mathbf{h}}_{2}, \dots, \bar{\mathbf{h}}_{K}] \in \mathbb{C}^{\bar{M}\bar{N} \times K}$ as the channel matrix between all candidate positions of the CL-MA array at the BS and all users.

To facilitate discrete position optimization, we define two binary matrices, $\mathbf{B}_{x} \in \mathbb{Z}^{M \times \bar{M}}$ and $\mathbf{B}_{y} \in \mathbb{Z}^{N \times \bar{N}}$, to indicate the selection of antenna positions over the horizontal and vertical directions, respectively. The $m$-th row, $1 \leq m \leq M$, of $\mathbf{B}_{x}$ has and only has one entry equal to $1$, which corresponds to the horizontal position of the $m$-th column of MAs. Similarly, the $n$-th row, $1 \leq n \leq N$, of $\mathbf{B}_{y}$ has and only has one entry equal to $1$, which corresponds to the vertical position of the $n$-th row of MAs. Therefore, the horizontal and vertical APVs can be determined by the binary selection matrices as $\mathbf{x} = \mathbf{B}_{x} \bar{\mathbf{x}}$ and $\mathbf{y} = \mathbf{B}_{y} \bar{\mathbf{y}}$, respectively. Meanwhile, the channel vector from user $k$ to all the MAs at the BS in \eqref{eq_channel} can be recast into
\begin{equation}\label{eq_channel_sele}
	\begin{aligned}
		&\mathbf{h}_{k}(\mathbf{B}_{x}\bar{\mathbf{x}}, \mathbf{B}_{y}\bar{\mathbf{y}})=(\mathbf{B}_{x} \bar{\mathbf{F}}_{k}^{\mathrm{hor}})^{\mathrm{H}} \odot (\mathbf{B}_{y} \bar{\mathbf{F}}_{k}^{\mathrm{ver}})^{\mathrm{H}} \mathbf{b}_{k}\\
		= &(\mathbf{B}_{x} \otimes \mathbf{B}_{y}) \left((\bar{\mathbf{F}}_{k}^{\mathrm{hor}})^{\mathrm{H}} \odot (\bar{\mathbf{F}}_{k}^{\mathrm{ver}})^{\mathrm{H}}\right) \mathbf{b}_{k} \\
		= &(\mathbf{B}_{x} \otimes \mathbf{B}_{y}) \bar{\mathbf{h}}_{k},~1 \leq k \leq K.
	\end{aligned}
\end{equation}
Therefore, the channel matrix between the CL-MA array at the BS and all users can be given by $(\mathbf{B}_{x} \otimes \mathbf{B}_{y}) \bar{\mathbf{H}} \in \mathbb{C}^{MN \times K}$.

Based on the above discretization, problem \eqref{eq_problem_APV} for antenna position optimization can be simplified as
\begin{subequations}\label{eq_problem_dis}
	\begin{align}
		\mathop{\min} \limits_{\mathbf{B}_{x}, \mathbf{B}_{y}}
		&P(\mathbf{B}_{x}, \mathbf{B}_{y}) \label{eq_problem_dis_a}\\
		\mathrm{s.t.}~  &\left[\mathbf{B}_{x}\right]_{m,\bar{m}} \in \{0,1\},~1 \leq m \leq M,~1 \leq \bar{m} \leq \bar{M}, \label{eq_problem_dis_b}\\
		&\sum \limits_{\bar{m}=1}^{\bar{M}}\left[\mathbf{B}_{x}\right]_{m,\bar{m}} = 1,~1 \leq m \leq M, \label{eq_problem_dis_c}\\
		&\left[\mathbf{B}_{x} \bar{\mathbf{x}}\right]_{m+1} - \left[\mathbf{B}_{x} \bar{\mathbf{x}}\right]_{m} \geq d_{\min,x},~1 \leq m \leq M-1, \label{eq_problem_dis_d}\\
		&\left[\mathbf{B}_{y}\right]_{n,\bar{n}} \in \{0,1\},~1 \leq n \leq N,~1 \leq \bar{n} \leq \bar{N}, \label{eq_problem_dis_e}\\
		&\sum \limits_{\bar{n}=1}^{\bar{N}}\left[\mathbf{B}_{y}\right]_{n,\bar{n}} = 1,~1 \leq n \leq N, \label{eq_problem_dis_f}\\
		&\left[\mathbf{B}_{y} \bar{\mathbf{y}}\right]_{n+1} - \left[\mathbf{B}_{y} \bar{\mathbf{y}}\right]_{n} \geq d_{\min,y},~1 \leq n \leq N-1, \label{eq_problem_dis_g}
	\end{align}
\end{subequations}
where $P(\mathbf{B}_{x}, \mathbf{B}_{y}) \triangleq \tilde{P}(\mathbf{B}_{x}\bar{\mathbf{x}}, \mathbf{B}_{y}\bar{\mathbf{y}})$ denotes the total transmit power as a function of the binary selection matrices. Different from conventional antenna selection (AS) system, problem \eqref{eq_problem_dis} involves the constraints \eqref{eq_problem_dis_d} and \eqref{eq_problem_dis_g} on the minimum inter-antenna spacing.
For this nonlinear combinatorial optimization problem, it generally requires an exponential complexity to obtain a globally optimal solution, which is computationally prohibitive for large $M$ and $\bar{M}$ (or $N$ and $\bar{N}$). Next, we develop a low-complexity iterative algorithm to obtain suboptimal solutions for problem \eqref{eq_problem_dis}. The key idea of the proposed algorithm is to conduct greedy search for the position of each row/column of the antennas, which consists of two phases: \emph{sequential elimination} and \emph{successive refinement}.

\begin{algorithm}[t]\footnotesize
	\caption{Proposed solution for problem \eqref{eq_problem_dis}.}
	\label{alg_APV}
	\begin{algorithmic}[1]
		\REQUIRE ~$M$, $N$, $\bar{M}$, $\bar{N}$, $\bar{\mathbf{H}}$, $\sigma^{2}$, $\{r_{k}\}$, $d_{\min,x}$, $d_{\min,y}$, $x_{\max}$, $y_{\max}$, $\epsilon$.
		\ENSURE ~$\mathbf{B}_{x}$ and $\mathbf{B}_{y}$. \\
		\STATE Initialize $\mathbf{B}_{x} \leftarrow \mathbf{I}_{\bar{M}}$ and $\mathbf{B}_{y} \leftarrow \mathbf{I}_{\bar{N}}$.  \% Sequential elimination
		\FOR   {$i=1:\max\{\bar{M}-M, \bar{N}-N\}$}
		\IF		{$i \leq \bar{M}-M$}
		\STATE Set $P_{\min} \leftarrow P(\mathbf{B}_{x}, \mathbf{B}_{y})$ according to \eqref{eq_total_power_min}.
		\STATE Set $\mathbf{B}_{x}^{i} \leftarrow \mathbf{B}_{x}$. \% of dimension $(\bar{M}-i+1) \times \bar{M}$
		\FOR	{$j = 1:\bar{M}-i+1$}
		\STATE Set $\mathbf{B}_{x}^{\mathrm{eli}}$ as a submatrix of $\mathbf{B}_{x}^{i}$ with the $j$-th row eliminated.
		\IF		{$P(\mathbf{B}_{x}^{\mathrm{eli}}, \mathbf{B}_{y}) < P_{\min}$}
		\STATE Update $\mathbf{B}_{x} \leftarrow \mathbf{B}_{x}^{\mathrm{eli}}$.
		\STATE Update $P_{\min} \leftarrow P(\mathbf{B}_{x}^{\mathrm{eli}}, \mathbf{B}_{y})$.
		\ENDIF
		\ENDFOR
		\ENDIF
		\IF		{$i \leq \bar{N}-N$}
		\STATE Set $P_{\min} \leftarrow P(\mathbf{B}_{x}, \mathbf{B}_{y})$ according to \eqref{eq_total_power_min}.
		\STATE Set $\mathbf{B}_{y}^{i} \leftarrow \mathbf{B}_{y}$. \% of dimension $(\bar{N}-i+1) \times \bar{N}$
		\FOR	{$j = 1:\bar{N}-i+1$}
		\STATE Set $\mathbf{B}_{y}^{\mathrm{eli}}$ as a submatrix of $\mathbf{B}_{y}^{i}$ with the $j$-th row eliminated.
		\IF		{$P(\mathbf{B}_{x}, \mathbf{B}_{y}^{\mathrm{eli}}) < P_{\min}$}
		\STATE Update $\mathbf{B}_{y} \leftarrow \mathbf{B}_{y}^{\mathrm{eli}}$.
		\STATE Update $P_{\min} \leftarrow P(\mathbf{B}_{x}, \mathbf{B}_{y}^{\mathrm{eli}})$.
		\ENDIF
		\ENDFOR
		\ENDIF
		\ENDFOR
		\STATE Set $\hat{P}_{\min} \leftarrow +\inf$.  \% Successive refinement
		\WHILE {$|P_{\min} - \hat{P}_{\min}| < \epsilon$} 
		\STATE Set $\hat{P}_{\min} \leftarrow P_{\min}$.
		\FOR   {$i=1:\max\{M, N\}$}
		\IF		{$i \leq M$}
		\STATE Set $P_{\min} \leftarrow +\inf$.
		\STATE Set $\mathbf{B}_{x}^{i} \leftarrow \mathbf{B}_{x}$. \% of dimension $M \times \bar{M}$
		\FOR	{$j = 1:\bar{M}$}
		\STATE Set $\mathbf{B}_{x}^{\mathrm{ref}}$ as $\mathbf{B}_{x}^{i}$ with its $i$-th row replaced by the $j$-th row of $\mathbf{I}_{\bar{M}}$.
		\IF		{$P(\mathbf{B}_{x}^{\mathrm{ref}}, \mathbf{B}_{y}) < P_{\min}$ and \eqref{eq_problem_dis_d} holds for $m=i$}
		\STATE Update $\mathbf{B}_{x} \leftarrow \mathbf{B}_{x}^{\mathrm{ref}}$.
		\STATE Update $P_{\min} \leftarrow P(\mathbf{B}_{x}^{\mathrm{ref}}, \mathbf{B}_{y})$.
		\ENDIF
		\ENDFOR
		\ENDIF
		\IF		{$i \leq N$}
		\STATE Set $P_{\min} \leftarrow +\inf$.
		\STATE Set $\mathbf{B}_{y}^{i} \leftarrow \mathbf{B}_{y}$. \% of dimension $N \times \bar{N}$
		\FOR	{$j = 1:\bar{N}$}
		\STATE Set $\mathbf{B}_{y}^{\mathrm{ref}}$ as $\mathbf{B}_{y}^{i}$ with its $i$-th row replaced by the $j$-th row of $\mathbf{I}_{\bar{N}}$.
		\IF		{$P(\mathbf{B}_{x}, \mathbf{B}_{y}^{\mathrm{ref}}) < P_{\min}$ and \eqref{eq_problem_dis_g} holds for $n=i$}
		\STATE Update $\mathbf{B}_{y} \leftarrow \mathbf{B}_{y}^{\mathrm{ref}}$.
		\STATE Update $P_{\min} \leftarrow P(\mathbf{B}_{x}, \mathbf{B}_{y}^{\mathrm{ref}})$.
		\ENDIF
		\ENDFOR
		\ENDIF
		\ENDFOR
		\ENDWHILE
		\RETURN $\mathbf{B}_{x}$ and $\mathbf{B}_{y}$.
	\end{algorithmic}
\end{algorithm}

As shown in Algorithm \ref{alg_APV}, we initialize $\mathbf{B}_{x}$ and $\mathbf{B}_{y}$ as identity matrices $\mathbf{I}_{\bar{M}}$ and $\mathbf{I}_{\bar{N}}$, respectively. The motivation of employing such an initialization and sequential elimination is to improve the global search efficiency for antenna position optimization, while ensuring the invertibility of $\mathbf{H}(\mathbf{B}_{x}\bar{\mathbf{x}}, \mathbf{B}_{y}\bar{\mathbf{y}})^{\mathrm{H}}\mathbf{H}(\mathbf{B}_{x}\bar{\mathbf{x}}, \mathbf{B}_{y}\bar{\mathbf{y}})$ for calculating the total transmit power in \eqref{eq_total_power_min} during the iterations. Note that the initialized $\mathbf{B}_{x}$ and $\mathbf{B}_{y}$ are not feasible to problem \eqref{eq_problem_dis}. To obtain feasible solutions, we should eliminate $(\bar{M}-M)$ rows from $\mathbf{B}_{x}$ and $(\bar{N}-N)$ rows from $\mathbf{B}_{y}$, respectively. From the physical point of view, $\mathbf{B}_{x}=\mathbf{I}_{\bar{M}}$ and $\mathbf{B}_{y}=\mathbf{I}_{\bar{N}}$ can be interpreted as a virtual array where all candidate positions are deployed with antennas, while we need to remove $(\bar{M}-M)$ columns and $(\bar{N}-N)$ rows of antennas such that only $M \times N$ antennas are present. To this end, we sequentially eliminate one row from $\mathbf{B}_{x}$ and $\mathbf{B}_{y}$ until only $M$ and $N$ rows remain, respectively. Since such an elimination process reduces the number of virtual antennas after each step, the total transmit power of users generally increases over the iterations. Thus, the greedy mechanism is further employed to guarantee the smallest increment in the objective function for each elimination, which are conducted in lines 3-13 for $\mathbf{B}_{x}$'s optimization and in lines 14-24 for $\mathbf{B}_{y}$'s optimization, respectively.

Note that the above greedy elimination may not guarantee the inter-antenna spacing constraints \eqref{eq_problem_dis_d} and \eqref{eq_problem_dis_g}. Besides, the objective function may be further decreased by adjusting the antenna positions. Thus, we further successively refine each row of the binary selection matrices with the other rows being fixed. In particular, the solution for the optimized row of $\mathbf{B}_{x}$ or $\mathbf{B}_{y}$ is selected from the candidate ones which yields the minimum value of the objective function while satisfying the inter-antenna spacing constraint with other antennas. We iteratively update $\mathbf{B}_{x}$ in lines 30-40 and $\mathbf{B}_{y}$ in lines 41-51 until the change in the objective function's value falls below a predefined threshold $\epsilon$. 

The convergence of Algorithm \ref{alg_APV} is guaranteed by the above greedy search over finite discrete solutions. Note that the objective function $P(\mathbf{B}_{x}, \mathbf{B}_{y})$ has a closed-form expression shown in \eqref{eq_total_power_min}. To further reduce the computational complexity during the iterations, we calculate it in line 8 using the following equations:
\begin{equation}
	\begin{aligned}
		& P(\mathbf{B}_{x}^{\mathrm{eli}}, \mathbf{B}_{y}) = \tilde{P}(\mathbf{B}_{x}^{\mathrm{eli}}\bar{\mathbf{x}}, \mathbf{B}_{y}\bar{\mathbf{y}})\\
		=&\mathrm{tr} \left\{ \left(\bar{\mathbf{H}}^{\mathrm{H}} (\mathbf{B}_{x}^{\mathrm{eli}} \otimes \mathbf{B}_{y})^{\mathrm{H}} (\mathbf{B}_{x}^{\mathrm{eli}} \otimes \mathbf{B}_{y}) \bar{\mathbf{H}} \right)^{-1}\mathbf{\Omega} \right\} \\
		=&\mathrm{tr} \left\{ \left(\mathbf{Z}_{x} - \bar{\mathbf{H}}^{\mathrm{H}} (\mathbf{b}_{x}^{i,j} \otimes \mathbf{B}_{y})^{\mathrm{H}} (\mathbf{b}_{x}^{i,j} \otimes \mathbf{B}_{y}) \bar{\mathbf{H}} \right)^{-1}\mathbf{\Omega} \right\},
	\end{aligned}
\end{equation}
where $\mathbf{Z}_{x}^{i} = \bar{\mathbf{H}}^{\mathrm{H}} (\mathbf{B}_{x}^{i} \otimes \mathbf{B}_{y})^{\mathrm{H}} (\mathbf{B}_{x}^{i} \otimes \mathbf{B}_{y}) \bar{\mathbf{H}}$ is a constant matrix for a given iteration index $i$ and $\mathbf{b}_{x}^{i,j}$ denotes the $j$-th row of $\mathbf{B}_{x}^{i}$. The calculation of $P(\mathbf{B}_{x}, \mathbf{B}_{y}^{\mathrm{eli}})$, $P(\mathbf{B}_{x}^{\mathrm{ref}}, \mathbf{B}_{y})$, and $P(\mathbf{B}_{x}, \mathbf{B}_{y}^{\mathrm{ref}})$ in lines 19, 35, and 46 can be conducted in a similar way. Then, the main algorithm complexity is determined by the total number of iterations in Algorithm \ref{alg_APV}. In lines 3-13 and lines 14-24, the total numbers of times for computing $P(\mathbf{B}_{x}, \mathbf{B}_{y})$ are in the order of $\mathcal{O}(\bar{M}^{2})$ and $\mathcal{O}(\bar{N}^{2})$, respectively. Denoting the maximum number of iterations for successive refinement as $T$, the total number of times for computing $P(\mathbf{B}_{x}, \mathbf{B}_{y})$ in lines 27-53 is in the order of $\mathcal{O}(T\bar{M}M+T\bar{N}N)$. Given the complexity of computing $P(\mathbf{B}_{x}, \mathbf{B}_{y})$ each time via matrix inversion as $\mathcal{O}(K^{3})$, the total computational complexity of Algorithm \ref{alg_APV} is thus $\mathcal{O}(K^{3}(\bar{M}^{2}+\bar{N}^{2}+T\bar{M}M+T\bar{N}N))$.

\section{Statistical Channel-based APV Optimization}
The antenna movement optimization in the previous section is based on instantaneous channels between the BS and users. In practice, as the users and/or environmental scatterers move, the AoAs and coefficients for the channel paths from each user to the BS may vary over time. Thus, the instantaneous channel vector for each user in \eqref{eq_channel} may exhibit random variation over time. In these scenarios, frequent antenna repositioning based on instantaneous channels results in high movement overhead and energy consumption, especially for high-mobility users undergoing fast-fading channels. 

To address this issue, we propose a two-timescale optimization strategy for the considered CL-MA array aided system. Specifically, the APVs are optimized based on the knowledge of statistical channels, which remain unchanged over a long time period unless the channel statistics/user distributions change. Besides, the transmit power of users and the receive combining matrix at the BS are designed based on instantaneous channels, which are similar to that in conventional multiuser communication systems aided by FPAs. Accordingly, the two-timescale optimization problem to minimize the expected transmit power of users can be formulated as
\begin{subequations}\label{eq_problem_stc}
	\begin{align}
		\mathop{\min} \limits_{\mathbf{x}, \mathbf{y}}~
		&\mathbb{E}_{\mathbf{H}}\left\{\mathop{\min} \limits_{\mathbf{p}, \mathbf{W}} ~ \sum \limits_{k=1}^{K} p_{k}\right\} \label{eq_problem_stc_a}\\
		\mathrm{s.t.}~~  &\eqref{eq_problem_ist_b},\eqref{eq_problem_ist_c},\eqref{eq_problem_ist_d},\eqref{eq_problem_ist_e},\eqref{eq_problem_ist_f},\eqref{eq_problem_ist_g}.
	\end{align}
\end{subequations}
Note that the statistical channel-based antenna position optimization is conducted offline to obtain APVs. In practice, once the CL-MA array has been configured/moved to the optimized positions, the parameter optimization (for $\mathbf{p}$ and $\mathbf{W}$) can be conducted in real time based on users' instantaneous channels (see Section IV-A).

\subsection{Performance Analysis}
Before solving problem \eqref{eq_problem_stc}, we analyze the lower bound on the expectation of the total transmit power in \eqref{eq_problem_stc_a}. According to Theorem \ref{theo_bound}, we have 
\begin{equation}\label{eq_power_bound_exp}
	\begin{aligned}
		&\mathop{\min} \limits_{\mathbf{x}, \mathbf{y}}~
		\mathbb{E}_{\mathbf{H}}\left\{\mathop{\min} \limits_{\mathbf{p}, \mathbf{W}} ~ \sum_{k=1}^{K} p_{k}\right\}\\
		&~~~~\geq \mathbb{E}_{\mathbf{H}}\left\{\mathop{\min} \limits_{\mathbf{x}, \mathbf{y}, \mathbf{p}, \mathbf{W}} ~ \sum_{k=1}^{K} p_{k}\right\} \geq \mathbb{E}_{\mathbf{H}}\left\{\sum_{k=1}^{K} \bar{p}_{k}\right\},
	\end{aligned}
\end{equation}
where $\bar{p}_{k}$ is the lower bound on the transmit power for user $k$ shown in \eqref{eq_bound_power}.

Next, we resort to analyzing the expectation of the lower bound on the total transmit power. Similar to Section III, we consider the case that the channel for each user is dominated by an LoS path. The amplitude of the LoS path coefficient is given by $|b_{k}|=\frac{\lambda}{4 \pi D_{k}^{\alpha/2}}$, $1 \leq k \leq K$, with $\alpha$ denoting the path loss exponent. $D_{k}$ is the distance between the BS and user $k$, which can be regarded as a random variable that impacts the channel gain as each user moves within the coverage sector of the BS. Thus, the expectation of $\bar{p}_{k}$ in \eqref{eq_bound_power} is given by 
\begin{equation}\label{eq_power_bound_expect}
	\begin{aligned}
		&\mathbb{E}_{\mathbf{H}}\{\bar{p}_{k}\} = \mathbb{E}_{D_{k}} \left\{ \frac{\sigma^{2}(2^{r_{k}-1})}{MN|b_{k}|^{2}}\right\} \\
		=&\mathbb{E}_{D_{k}} \left\{D_{k}^{\alpha}\right\} \frac{16 \pi^{2}  \sigma^{2}(2^{r_{k}-1})}{MN\lambda^{2}},~1 \leq k \leq K.
	\end{aligned}
\end{equation}
Since multiple users are randomly distributed and move independently in the site, it is reasonable to assume $\{D_{k}\}_{k=1}^{K}$ as independent and identically distributed (i.i.d.) random variables. For notation simplicity, we denote $D$ as the random variable representing the distance between the BS and a user, which follows the identical distribution as $D_{k}$, $1 \leq k \leq K$. As a result, the expectation of the lower bound in \eqref{eq_power_bound_exp} is given by
\begin{equation}\label{eq_total_power_expect}
	\begin{aligned}
		&\mathbb{E}_{\mathbf{H}}\left\{ \sum \limits_{k=1}^{K} \bar{p}_{k}\right\} = \mathbb{E}_{D} \left\{D^{\alpha}\right\} \sum_{k=1}^{K} \frac{16 \pi^{2}  \sigma^{2}(2^{r_{k}-1})}{MN\lambda^{2}},
	\end{aligned}
\end{equation}
which decreases with the total number of antenna $MN$, but increases with the expectation of the powered distance $\mathbb{E}_{D} \left\{D^{\alpha}\right\}$, the number of users $K$, and the rate requirement for each user $r_{k}$. The average distance between the BS and user depends on both the deployment of the BS and the distribution of users. For example, assume that the CL-MA array of the BS is fixed at height $H$, while the users are randomly distributed within a sector on the ground following 2D uniform distribution. Denote the maximal and minimal 2D distance between the BS and user as $D_{\max}$ and $D_{\min}$, respectively. It is easy to show that $D^{2}$ follows uniform distribution within interval $[H^{2}+D_{\min}^{2}, H^{2}+D_{\max}^{2}]$. For free-space propagation loss with $\alpha=2$, we have $\mathbb{E}_{D} \left\{D^{2}\right\} = H^{2} + (D_{\min}^{2} + D_{\max}^{2})/2$.

\subsection{Optimization Algorithm}
According to Theorem \ref{theo_APV}, the optimal APVs for achieving the lower bound on the total transmit power vary with different instantaneous channels. Thus, it is generally difficult to achieve the lower bound in \eqref{eq_total_power_expect} via antenna position optimization based on the knowledge of statistical channels. In view of this, we should develop a dedicated algorithm to solve problem \eqref{eq_problem_stc}.
Different from the instantaneous channel-based design in \eqref{eq_problem_ist}, the expectation of the transmit power in \eqref{eq_problem_stc_a} does not have an explicit expression, which results in the APV optimization outside the expectation intractable. To facilitate antenna position optimization, we resort to approximating the expectation of the total transmit power by the method of Monte Carlo simulation. 

Specifically, a sufficiently large number of instantaneous channels are randomly generated based on the given statistical channel distribution of users. Denote the $s$-th random realization of instantaneous channel mapping as $\mathbf{H}^{(s)}$, $1 \leq s \leq S$, with $S$ being the total number of realizations. For each instantaneous channel realization, we can also adopt the ZF-based receive combining in \eqref{eq_ZF} and transmit power in \eqref{eq_power_min}, where the corresponding total transmit power in \eqref{eq_total_power_min} for the $s$-th channel realization can be obtained as $\tilde{P}(\mathbf{x}, \mathbf{y}; \mathbf{H}^{(s)})$.
Then, the expectation is approximated by the average value of the total transmit power of users over all $S$ instantaneous channel realizations as
\begin{equation}\label{eq_Monte_Carlo}
	\mathbb{E}_{\mathbf{H}} \left\{\tilde{P}(\mathbf{x}, \mathbf{y}; \mathbf{H})\right\}
	\approx \frac{1}{S} \sum \limits_{s=1}^{S} \tilde{P}(\mathbf{x}, \mathbf{y}; \mathbf{H}^{(s)}).
\end{equation}
For sufficiently large $S$, the statistical information of wireless channels between the BS and users, e.g., the distributions of users and scatterers in the site, can be effectively captured by \eqref{eq_Monte_Carlo}. As such, \eqref{eq_Monte_Carlo} can provide a good approximation of the objective function with the expectation over channel distributions. Based on the above derivation, problem \eqref{eq_problem_stc} can be efficiently solved by Algorithm \ref{alg_APV}, whereas the objective function is replaced with \eqref{eq_Monte_Carlo}.

\section{Performance Evaluation}
This section presents the simulation results to evaluate the performance of our proposed CL-MA array for minimizing the transmit power of users. The size of the CL-MA array is $M \times N = 6 \times 6$, which is installed on the BS at the altitude of $10$ meters (m). The region size for antenna movement is $20\lambda \times 20\lambda$, which is discretized by a resolution of $\lambda/4$ in Algorithm \ref{alg_APV}. The BS covers a sector spanned of $120^{\circ}$ azimuth angle. In the sector, three cubical buildings are assumed to be located at $[15, 0, 30]$, $[0, 0, 40]$, and $[-15, 0, 30]$ m, with the height of $30$ m and the bottom side length of $10$ m. The total number of users is $K=18$. In particular, half of the users are randomly distributed on the ground with the distance to the BS ranging from $5$ to $50$ m. The other half of users are randomly distributed in the buildings. The carrier frequency is $30$ GHz. The LoS channel between the BS and each user is considered with the amplitude of the path coefficient given by $|b_{k}|=\frac{\lambda}{4\pi D_{k}}$, where $D_{k}$ is the distance between the BS and user $k$. The noise power is $-80$ dBm and the minimum rate requirement for each user is $r_{k}=3$ bps/Hz. The total number of Monte Carlo simulations for user locations and instantaneous channel realizations is $S=1000$.

In addition to the proposed CL-MA designs based on instantaneous and statistical channels. Three benchmark schemes defined as follows are considered for performance comparison. The ``element-wise MA'' scheme performs independent antenna movement within the 2D rectangular region, where the antenna positions are optimized by element-wise sequential elimination and successive refinement similar to that in Algorithm \ref{alg_APV}. The ``dense UPA'' scheme adopts uniform planar array (UPA) with inter-antenna spacing of $\lambda/2$. The ``sparse UPA'' scheme adopts UPA with inter-antenna spacing of $4\lambda$, which guarantees that the antennas are uniformly distributed over the entire 2D rectangular region. For a fair comparison, all benchmark schemes employ the same number of antennas as that of the CL-MA array, while the receive combining and transmit power are designed based on instantaneous channels using the method in Section IV-A. In addition, the globally lower bound on the total transmit power is given by \eqref{eq_total_power_expect}.

\begin{figure}[t]
	\begin{center}
		\includegraphics[width=\figwidth cm]{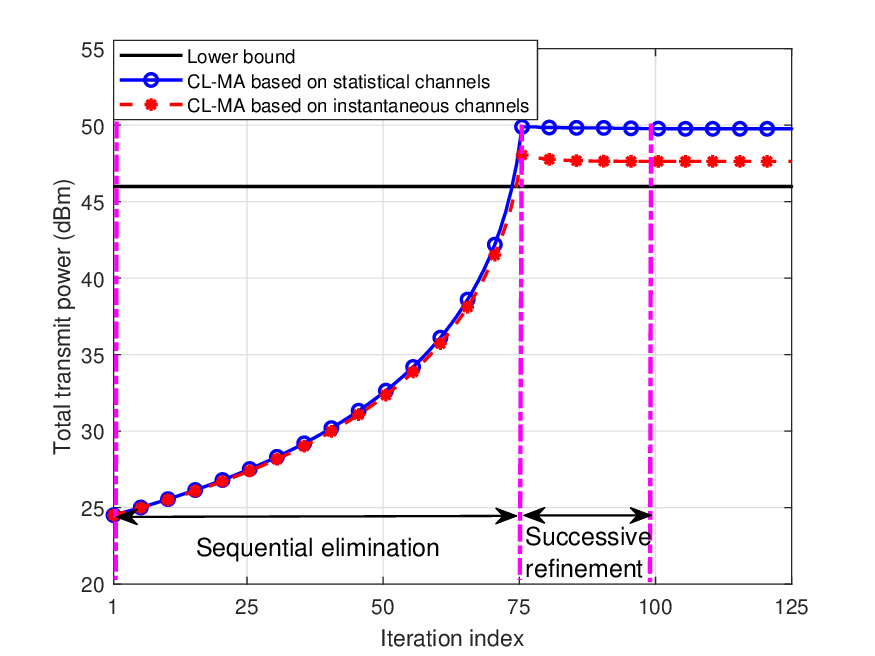}
		\caption{Convergence evaluation of the proposed Algorithm \ref{alg_APV}.}
		\label{Fig_Iteration}
	\end{center}
\end{figure}
First, we evaluate the convergence of the proposed Algorithm \ref{alg_APV} in Fig. \ref{Fig_Iteration}. As can be observed, the sequential elimination phase completes within $75$ iterations. In this phase, the total transmit power of the proposed CL-MA solutions increase with the iteration because the virtual antennas are removed for each iteration. Then, in the successive refinement phase, the total transmit power slightly decreases over the iterations because the adjustment of antenna positions. This phase spends about $25$ iterations, which demonstrates the quick convergence of the proposed algorithm. In addition, we can observe that the proposed CL-MA scheme based on instantaneous channels can closely approach the globally lower bound on the total transmit power, highlighting the superiority of the proposed solutions. Moreover, the performance gap between the statistical and instantaneous channel-based CL-MA schemes is about $2$ dB. However, the statistical channel-based CL-MA scheme circumvents the frequent movement of antennas as the users randomly move within the coverage area of the BS. The overhead and energy consumption associated with antenna movement are significantly reduced, which thus offers a promising and viable solution for implementing CL-MA array-enabled BSs.

\begin{figure}[t]
	\begin{center}
		\includegraphics[width=\figwidth cm]{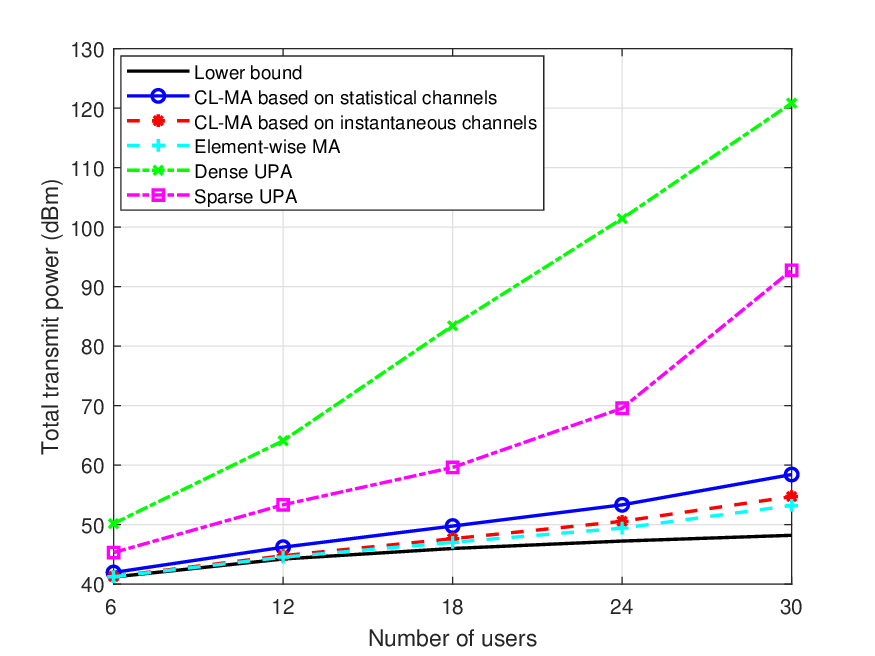}
		\caption{Comparison of total transmit power with the proposed and benchmark schemes versus the number of users.}
		\label{Fig_K}
	\end{center}
\end{figure}
Next, we compare the performance of the proposed and benchmark schemes. As shown in Fig. \ref{Fig_K}, the total transmit powers of different MA and FPA schemes all increase with the number of users. However, the performance of the MA schemes is much better than that of both FPA (i.e., dense and sparse UPA) schemes, especially for a larger number of users, e.g., over 30 dB power-saving for $K=30$. This is because MA arrays can efficiently decrease the channel correlation among users through antenna position optimization. The ZF-based receive combining will not cause a significant loss of the received signal power at the BS, thus yielding a total transmit power close to the globally lower bound. In contrast, due to the high correlation of multiple users' channels, the dense and sparse UPA schemes suffer from a severe loss in the desired signal power when nulling the interference. Besides, compared to the element-wise antenna movement, the proposed CL-MA architecture has a very small performance loss, with no larger than 1 dB increase in the transmit power. This demonstrates the superiority of the CL-MA array in reducing the hardware complexity yet maintaining high communication performance.

\begin{figure}[t]
	\begin{center}
		\includegraphics[width=\figwidth cm]{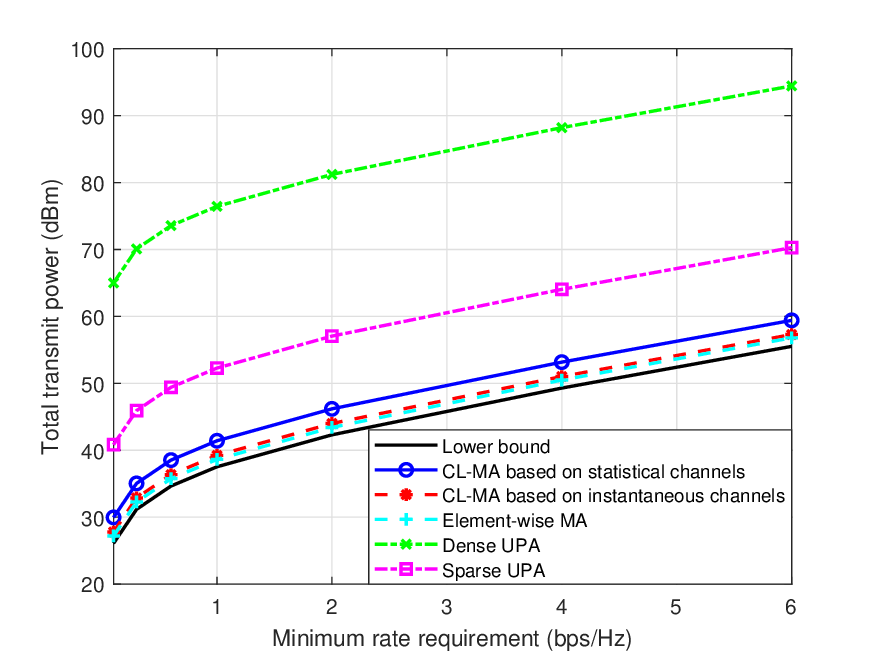}
		\caption{Comparison of total transmit power with the proposed and benchmark schemes versus the rate requirement of users.}
		\label{Fig_r}
	\end{center}
	\vspace{-12 pt}
\end{figure}

\begin{figure}[t]
	\begin{center}
		\includegraphics[width=\figwidth cm]{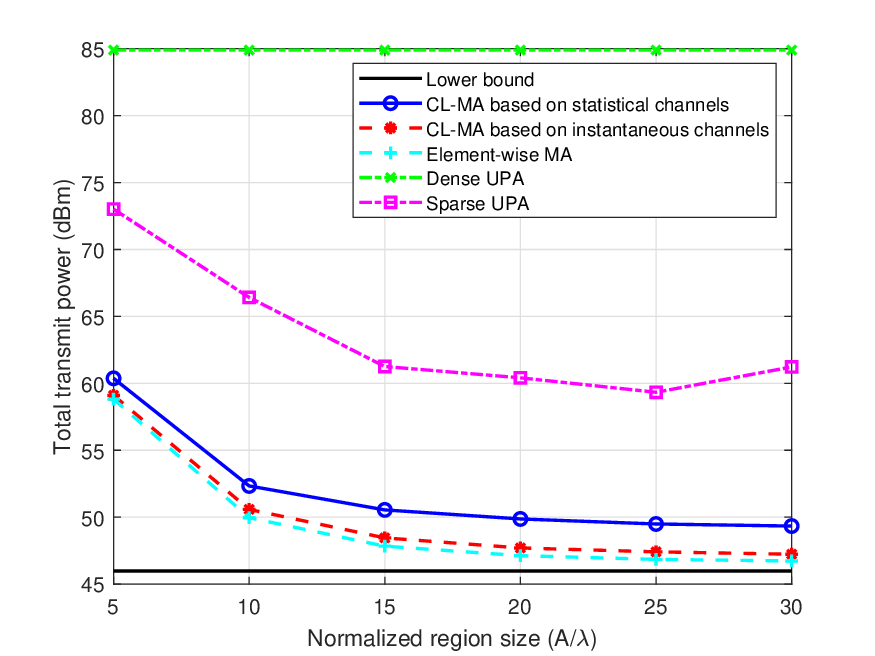}
		\caption{Comparison of total transmit power with the proposed and benchmark schemes versus the size of the antenna moving region.}
		\label{Fig_A}
	\end{center}
\end{figure}

\begin{figure}[t]
	\begin{center}
		\includegraphics[width=\figwidth cm]{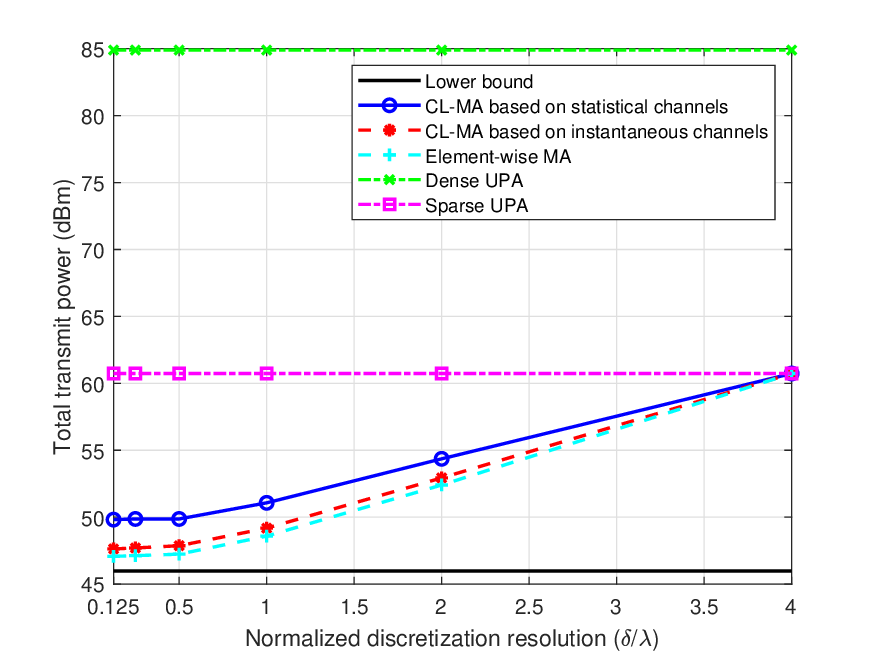}
		\caption{Comparison of total transmit power with the proposed and benchmark schemes versus the resolution of discretizing antenna moving region.}
		\label{Fig_sample}
	\end{center}
	\vspace{-12 pt}
\end{figure}

\begin{figure}[t]
	\begin{center}
		\includegraphics[width=\figwidth cm]{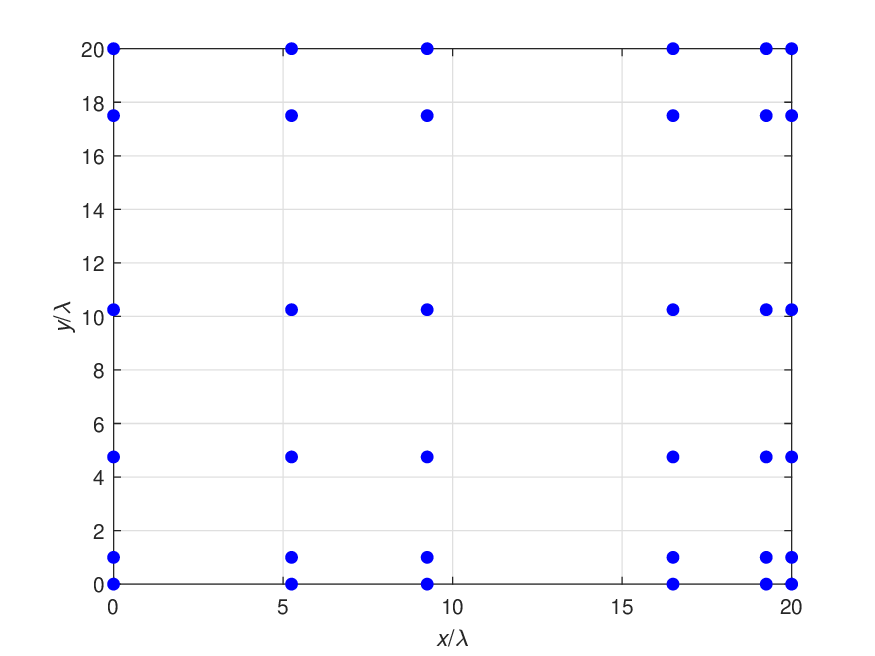}
		\caption{Illustration of the optimized geometry of the CL-MA array.}
		\label{Fig_Geo}
	\end{center}
\end{figure}
Fig. \ref{Fig_r} shows the total transmit power of users with varying rate requirements. It can be observed that the CL-MA scheme based on instantaneous channels can closely approach the lower bound on the total transmit power, where the performance gaps to the element-wise MA scheme and the lower bound are $0.5$ dB and $1.7$ dB, respectively. Due to the reduced DoFs in antenna movement, the CL-MA scheme based on statistical channels suffers from $2.2$ dB increase in the total transmit power compared to that based on instantaneous channels. However, it yields performance gains of $10.9$ dB and $35$ dB in saving power over the sparse UPA and dense UPA schemes, respectively. The results confirm again that the proposed CL-MA scheme can significantly reduce the hardware complexity and the statistical channel-based antenna position optimization can significantly reduce the movement overhead, while maintaining superior performance over conventional FPA systems.

\begin{figure*}[t]
	\centering
	\subfigure[CL-MA array]{\includegraphics[width=5.5 cm]{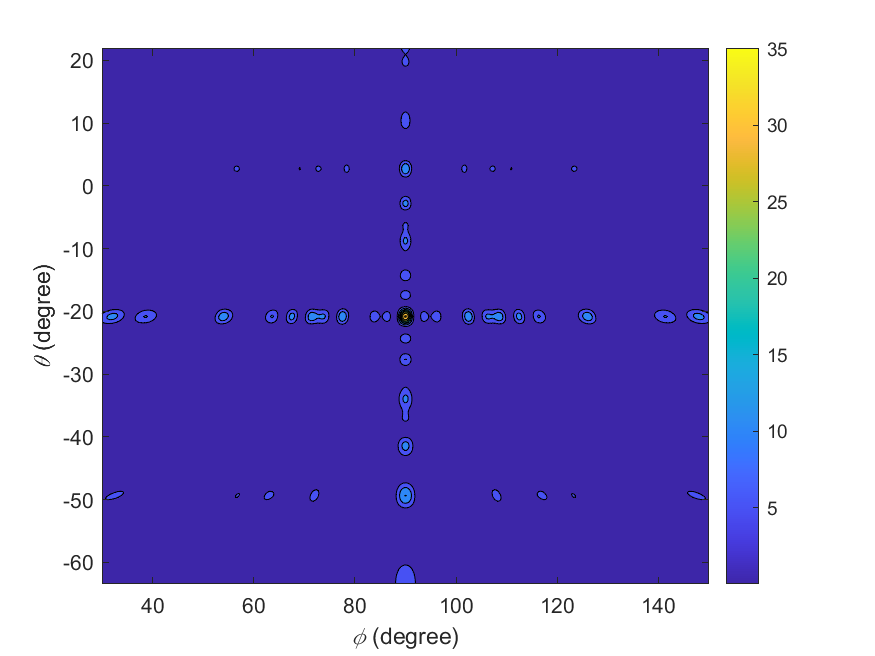} \label{Fig_Beam_MA}}
	\subfigure[Dense UPA]{\includegraphics[width=5.5 cm]{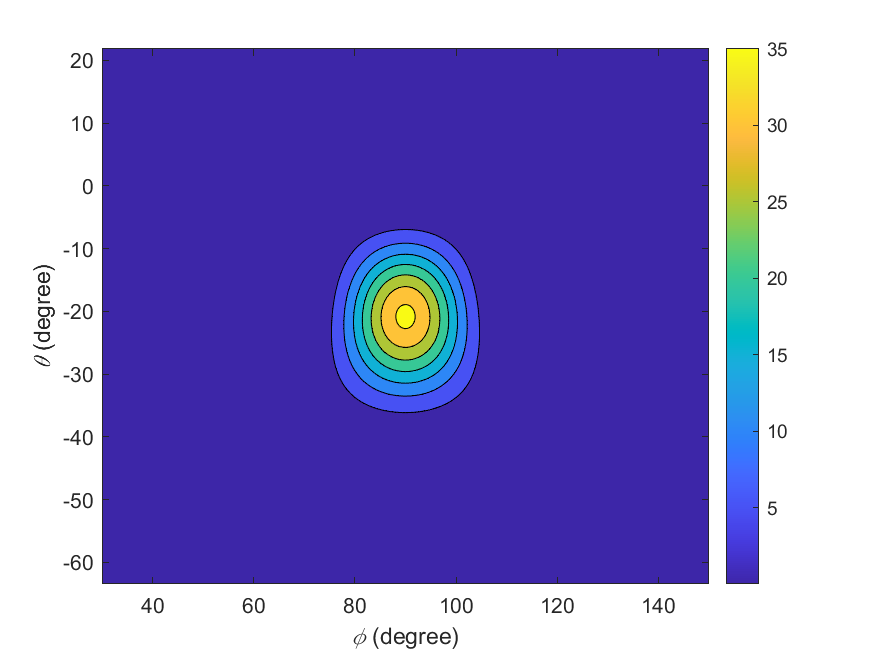} \label{Fig_Beam_dense}}
	\subfigure[Sparse UPA]{\includegraphics[width=5.5 cm]{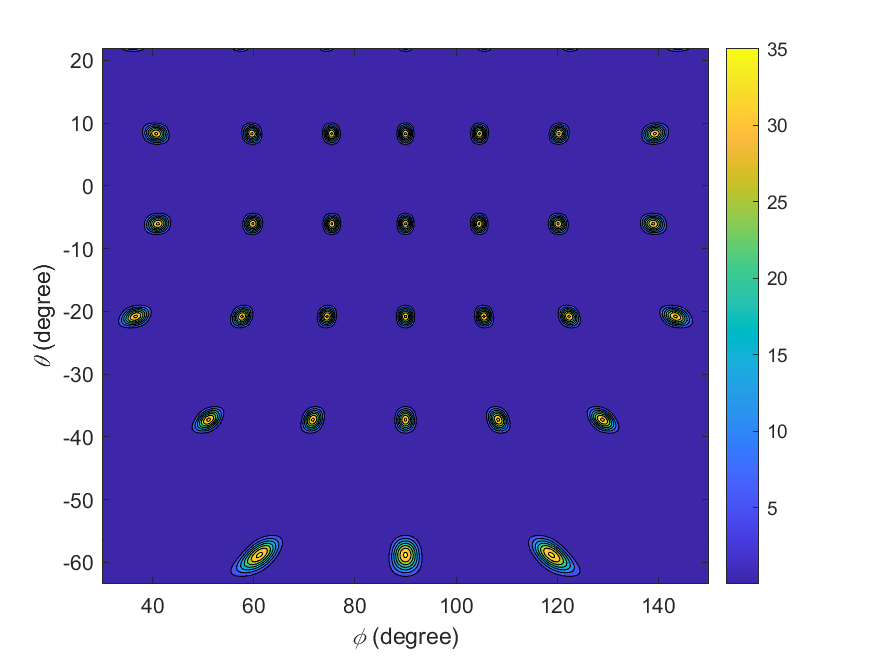} \label{Fig_Beam_sparse}}
	\caption{Beam patterns of different antenna arrays within the coverage sector of the BS.}
	\label{Fig_Beam}
	\vspace{-12 pt}
\end{figure*}

In Figs. \ref{Fig_A} and \ref{Fig_sample}, we evaluate the performance of our proposed CL-MA schemes versus the size and the discretization resolution of the antenna moving region, respectively. As the region size increases from $5\lambda$ to $20\lambda$, the MA systems achieve a rapid decrease in the total transmit power. For $A>20\lambda$, further enlarging the region size yields negligible improvements in system performance. This is because the channel spatial correlation of multiple users with optimized antenna positions is already sufficiently small for $A=20\lambda$. Besides, we can observe from Fig. \ref{Fig_sample} that the discretization resolution of the antenna moving region, i.e., $\delta$ in Algorithm \ref{alg_APV}, also significantly impacts the system performance. As $\delta$ decreases from $4\lambda$ to $0.5\lambda$, the total transmit power (in dB) decreases almost linearly. However, for $\delta<\lambda/2$, a higher discretization resolution does not contribute much to improving the system performance. This is because the beam pattern of an uniform array with inter-antenna spacing larger than half-wavelength has obstinate sidelobes \cite{TseFundaWC}, as shown in Fig. \ref{Fig_Beam_sparse}. This indicates that the users located within the sidelobes of each other have a high channel correlation, which deteriorates multiuser communication performance. The result in Fig. \ref{Fig_sample} suggests that a moderate discretization resolution, e.g., $\delta=\lambda/2$ or $\lambda/4$, suffices to achieve a near-optimal solution for the proposed algorithm of discrete antenna position optimization.

Finally, we show in Fig. \ref{Fig_Geo} the geometry of the CL-MA array optimized based on statistical channels. On one hand, the antennas are spread over the entire movable region, which maximizes the array aperture and minimizes the channel correlation of multiple users on average. On the other hand, the antennas of the CL-MA array are distributed in a non-uniform layout, which can mitigate the sidelobes of its beam pattern based on the users' locations. The beam patterns of the CL-MA array, dense UPA, and spare UPA are illustrated in Fig. \ref{Fig_Beam}, where the ranges of the elevation and azimuth angles are calculated according to the user distribution. It can be observed that the optimized CL-MA array has the narrowest main lobe and low sidelobes, which can efficiently decrease the channel correlation of users located within the coverage region of the BS. In comparison, the dense UPA has a wide main lobe, while the sparse UPA has numerous grating lobes. As a result, both schemes may lead to a high channel correlation between users, resulting in worse multiuser communication performance.

\section{Conclusions}
In this paper, we proposed a novel CL-MA architecture to reduce the hardware complexity associated with conventional element-wise MA systems. In particular, the CL-MA array was applied at BSs to enhance multiuser communication performance by jointly optimizing the APVs, receive combining matrix, and users’ transmit power to minimize the total uplink transmit power. We first analyzed the globally lower bound on the total transmit power and derived optimal solutions for the horizontal and vertical APVs in closed form under the condition of a single channel path for each user. For the general scenario with multiple channel paths of users, we developed a low-complexity algorithm employing a greedy approach for discrete antenna position optimization. Furthermore, to reduce antenna movement overhead, we proposed a statistical channel-based optimization strategy, enabling effective antenna position configuration without frequent repositioning under time-varying channel conditions. Simulation results validated the effectiveness of the proposed CL-MA schemes, demonstrating their ability to closely approach the theoretical power lower bound while significantly outperforming conventional FPA systems. Notably, the statistical channel-based antenna position optimization can achieve an appealing performance for CL-MA arrays with significantly reduced antenna movement overhead, thus offering a practically efficient solution for MA system implementation.

\appendices
\section{Proof of Theorem \ref{theo_bound}} \label{App_bound}
The SINR for user $k$ is lower-bounded by its received SNR by neglecting the interference in \eqref{eq_SINR}, i.e.,
\begin{equation}\label{eq_bound_SINR}
	\gamma_{k} \leq \bar{\gamma}_{k} = \frac{|\mathbf{w}_{k}^{\mathrm{H}}\mathbf{h}_{k}(\mathbf{x}, \mathbf{y})|^{2}p_{k}}{\|\mathbf{w}_{k}\|_{2}^{2}\sigma^{2}}.
\end{equation}
Substituting \eqref{eq_bound_SINR} into \eqref{eq_problem_ist_b}, we have
\begin{equation}\label{eq_bound_rate}
	\begin{aligned}
		&\log_{2}(1+\bar{\gamma}_{k}) \geq \log_{2}(1+\gamma_{k}) \geq r_{k}\\
		\Rightarrow &\frac{|\mathbf{w}_{k}^{\mathrm{H}}\mathbf{h}_{k}(\mathbf{x}, \mathbf{y})|^{2}p_{k}}{\|\mathbf{w}_{k}\|_{2}^{2}\sigma^{2}} \geq 2^{r_{k}}-1\\
		\Rightarrow & p_{k} \geq \frac{\|\mathbf{w}_{k}\|_{2}^{2} \sigma^{2}(2^{r_{k}-1})}{|\mathbf{w}_{k}^{\mathrm{H}}\mathbf{h}_{k}(\mathbf{x}, \mathbf{y})|^{2}},~1 \leq k \leq K.
	\end{aligned}
\end{equation}
Notice that $\frac{\|\mathbf{w}_{k}\|_{2}^{2}}{|\mathbf{w}_{k}^{\mathrm{H}}\mathbf{h}_{k}(\mathbf{x}, \mathbf{y})|^{2}}$ is minimized if the maximal ratio combining (MRC) is adopted, i.e., $\mathbf{w}_{k} = \mathbf{h}_{k}(\mathbf{x}, \mathbf{y})$. Thus, we have 
\begin{equation}\label{eq_bound_MRC}
	\frac{\|\mathbf{w}_{k}\|_{2}^{2} \sigma^{2}(2^{r_{k}-1})}{|\mathbf{w}_{k}^{\mathrm{H}}\mathbf{h}_{k}(\mathbf{x}, \mathbf{y})|^{2}} 
	\geq \frac{\sigma^{2}(2^{r_{k}-1})}{\|\mathbf{h}_{k}(\mathbf{x}, \mathbf{y})\|_{2}^{2}}.
\end{equation}
Furthermore, according to \eqref{eq_channel}, if each row in $\mathbf{F}_{k}^{\mathrm{hor}}(\mathbf{x})^{\mathrm{H}} \odot \mathbf{F}_{k}^{\mathrm{ver}}(\mathbf{y})^{\mathrm{H}}$ and $\mathbf{b}_{k}$ have aligned phases, an upper bound on the channel power gain for user $k$ is obtained as
\begin{equation}\label{eq_bound_channel_power}
	\begin{aligned}
		&\|\mathbf{h}_{k}(\mathbf{x}, \mathbf{y})\|_{2}^{2} = \|\mathbf{F}_{k}^{\mathrm{hor}}(\mathbf{x})^{\mathrm{H}} \odot \mathbf{F}_{k}^{\mathrm{ver}}(\mathbf{y})^{\mathrm{H}} \mathbf{b}_{k}\|_{2}^{2}\\
		\leq &\left\|\mathbf{1}_{MN} \|\mathbf{b}_{k}\|_{1}\right\|_{2}^{2} = MN\|\mathbf{b}_{k}\|_{1}^{2}.
	\end{aligned}
\end{equation}
Combining \eqref{eq_bound_rate}, \eqref{eq_bound_MRC}, and \eqref{eq_bound_channel_power}, we have $p_{k} \geq \bar{p}_{k}$. This thus completes the proof.

\section{Proof of Theorem \ref{theo_bound_condition}} \label{App_bound_condition}
Note that the CVO and MCP conditions on the channel vectors are sufficient to guarantee achieving the lower bound in \eqref{eq_bound_power}. For the case of a single channel path for each user shown in \eqref{eq_channel_LoS}, the MCP condition is naturally satisfied for all users, i.e., $\|\hat{\mathbf{h}}_{k}(\mathbf{x}, \mathbf{y})\|_{2}^{2} =  MN|b_{k}|^{2}$ always holds. Thus, in the following, we only need to prove that the CVO condition can always be satisfied for $K(K-1)/2 \leq I_{M}+I_{N}$. 

Denote $\mathcal{P}=\{(k,q) \arrowvert 1 \leq k < q \leq K\}$ as the set of all unordered user pairs. Next, we define $\mathcal{P}_{1}$ and $\mathcal{P}_{2}$ as two complementary sets w.r.t. the universal set $\mathcal{P}$, i.e., $\mathcal{P}_{1} \cup \mathcal{P}_{2} = \mathcal{P}$ and $\mathcal{P}_{1} \cap \mathcal{P}_{2} = \varPhi$, and thus we have $|\mathcal{P}_{1}|+|\mathcal{P}_{2}|=|\mathcal{P}|=K(K-1)/2$. If $K(K-1)/2 \leq I_{M}+I_{N}$ holds, we can always find two complementary sets satisfying $|\mathcal{P}_{1}| \leq I_{M}$ and $|\mathcal{P}_{2}| \leq I_{N}$. For example, we can choose $I_{M}$ elements in $\mathcal{P}$ to construct $\mathcal{P}_{1}$ and the other $K(K-1)/2 - I_{M}$ elements to construct $\mathcal{P}_{2}$.

It has been proved in \cite[Theorem 1]{zhu2023MAarray} that for a linear MA array of size $M$, the \emph{steering vector orthogonality (SVO) condition} can be satisfied for any $I$ ($\leq I_{M}$) pairs of users with different steering angles by optimizing the APV. In other words, for $|\mathcal{P}_{1}| \leq I_{M}$, there always exists an APV, $\mathbf{x}^{\star}$, satisfying $\mathbf{a}_{k}^{\mathrm{hor}}(\mathbf{x}^{\star})^{\mathrm{H}}\mathbf{a}_{q}^{\mathrm{hor}}(\mathbf{x}^{\star})=0$, $\forall (k,q) \in \mathcal{P}_{1}$. Similarly, for $|\mathcal{P}_{2}| \leq I_{N}$, there always exists an APV, $\mathbf{y}^{\star}$, satisfying $\mathbf{a}_{k}^{\mathrm{ver}}(\mathbf{y}^{\star})^{\mathrm{H}}\mathbf{a}_{q}^{\mathrm{ver}}(\mathbf{y}^{\star})=0$, $\forall (k,q) \in \mathcal{P}_{2}$. 
According to the channel vector in \eqref{eq_channel_LoS}, we have
\begin{equation}\label{eq_CVO}
	\begin{aligned}
		&\hat{\mathbf{h}}_{k}(\mathbf{x}^{\star}, \mathbf{y}^{\star})^{\mathrm{H}}\hat{\mathbf{h}}_{q}(\mathbf{x}^{\star}, \mathbf{y}^{\star}) \\
		= & b_{k}^{*}b_{q} \left(\mathbf{a}_{k}^{\mathrm{hor}}(\mathbf{x}^{\star}) \otimes \mathbf{a}_{k}^{\mathrm{ver}}(\mathbf{y}^{\star})\right)^{\mathrm{H}}\left(\mathbf{a}_{q}^{\mathrm{hor}}(\mathbf{x}^{\star}) \otimes \mathbf{a}_{q}^{\mathrm{ver}}(\mathbf{y}^{\star})\right)\\
		= & b_{k}^{*}b_{q} \mathbf{a}_{k}^{\mathrm{hor}}(\mathbf{x}^{\star})^{\mathrm{H}}\mathbf{a}_{q}^{\mathrm{hor}}(\mathbf{x}^{\star}) \times \mathbf{a}_{k}^{\mathrm{ver}}(\mathbf{y}^{\star})^{\mathrm{H}}\mathbf{a}_{q}^{\mathrm{ver}}(\mathbf{y}^{\star}),
	\end{aligned}
\end{equation}
which is always equal to zero because either $\mathbf{a}_{k}^{\mathrm{hor}}(\mathbf{x}^{\star})^{\mathrm{H}}\mathbf{a}_{q}^{\mathrm{hor}}(\mathbf{x}^{\star})=0$ or $\mathbf{a}_{k}^{\mathrm{ver}}(\mathbf{y}^{\star})^{\mathrm{H}}\mathbf{a}_{q}^{\mathrm{ver}}(\mathbf{y}^{\star})=0$ holds for any $(k,q) \in \mathcal{P}=\mathcal{P}_{1} \cup \mathcal{P}_{2}$. This thus completes the proof.

\section{Proof of Theorem \ref{theo_APV}} \label{App_APV}
To prove Theorem \ref{theo_APV}, we only need to verify that $\mathbf{x}^{\star}$ and $\mathbf{y}^{\star}$ can always guarantee the CVO condition, i.e., $\hat{\mathbf{h}}_{k}(\mathbf{x}^{\star}, \mathbf{y}^{\star})^{\mathrm{H}}\hat{\mathbf{h}}_{q}(\mathbf{x}^{\star}, \mathbf{y}^{\star})=0$, $\forall (k,q) \in \mathcal{P}=\mathcal{P}_{1} \cup \mathcal{P}_{2}$. For ease of exposition, we denote the $m$-th entry of $\mathbf{x}^{\star}$ as $x[m]$. According to \eqref{eq_dis_x}, we have 
\begin{equation}\label{eq_SVO_i}
	\begin{aligned}
		&\sum \limits_{u_{i}=0}^{m_{i}-1} \e^{\jj\frac{2\pi}{\lambda}u_{i} [\mathbf{d}_{x}]_{i}(\vartheta_{k_{1}^{(i)}}-\vartheta_{q_{1}^{(i)}})}\\
		=&\sum \limits_{u_{i}=0}^{m_{i}-1} \e^{\pm \jj2\pi u_{i} (\rho_{i}+1/m_{i})} = \sum \limits_{u_{i}=0}^{m_{i}-1} \e^{\pm \jj2\pi u_{i}/m_{i}} = 0.\\
	\end{aligned}
\end{equation}
For any $(k,q) \in \mathcal{P}_{1}$, denoting it as the $i$-th element without loss of generality, i.e., $(k_{1}^{(i)},q_{1}^{(i)})$, $1 \leq i \leq |\mathcal{P}_{1}|$, we thus have
\begin{equation}\label{eq_SVO_hor}
	\begin{aligned}
		&\mathbf{a}_{k_{1}^{(i)}}^{\mathrm{hor}}(\mathbf{x}^{\star})^{\mathrm{H}}\mathbf{a}_{q_{1}^{(i)}}^{\mathrm{hor}}(\mathbf{x}^{\star})
		=\sum \limits_{m=1}^{M} \e^{\jj\frac{2\pi}{\lambda}x[m](\vartheta_{k_{1}^{(i)}}-\vartheta_{q_{1}^{(i)}})}\\
		\overset{\text{(a)}}{=}&\sum \limits _{u_{I_{M}}=0}^{m_{I_{M}}-1}  \cdots \sum \limits_{u_{i+1}=0}^{m_{i+1}-1}  
		\sum \limits_{u_{i}=0}^{m_{i}-1} \sum \limits_{\hat{m}=1}^{M_{i}} \\
		&~~~~\e^{\jj\frac{2\pi}{\lambda}x[\hat{m}+u_{i} M_{i}+u_{i+1} M_{i+1}+\cdots+u_{I_{M}} M_{I_{M}}](\vartheta_{k_{1}^{(i)}}-\vartheta_{q_{1}^{(i)}})}\\
		\overset{\text{(b)}}{=}&\sum \limits _{u_{I_{M}}=0}^{m_{I_{M}}-1}  \cdots \sum \limits_{u_{i+1}=0}^{m_{i+1}-1}  
		\sum \limits_{u_{i}=0}^{m_{i}-1} \sum \limits_{\hat{m}=1}^{M_{i}} \\
		&~~~~\e^{\jj\frac{2\pi}{\lambda}\left(x[\hat{m}]+[\mathbf{d}_{x}]_{i}+[\mathbf{d}_{x}]_{i+1}+\cdots+[\mathbf{d}_{x}]_{I_{M}}\right)(\vartheta_{k_{1}^{(i)}}-\vartheta_{q_{1}^{(i)}})}\\
		\overset{\text{(c)}}{=}&\sum \limits _{u_{I_{M}}=0}^{m_{I_{M}}-1} \e^{\jj\frac{2\pi}{\lambda}u_{I_{M}} [\mathbf{d}_{x}]_{I_{M}}(\vartheta_{k_{1}^{(i)}}-\vartheta_{q_{1}^{(i)}})} \times \cdots \\
		&\times \sum \limits_{u_{i+1}=0}^{m_{i+1}-1}  \e^{\jj\frac{2\pi}{\lambda}u_{i+1} [\mathbf{d}_{x}]_{i+1}(\vartheta_{k_{1}^{(i)}}-\vartheta_{q_{1}^{(i)}})}\\
		&\times \sum \limits_{u_{i}=0}^{m_{i}-1} \e^{\jj\frac{2\pi}{\lambda}u_{i} [\mathbf{d}_{x}]_{i}(\vartheta_{k_{1}^{(i)}}-\vartheta_{q_{1}^{(i)}})}\\
		&\times\sum \limits_{\hat{m}=1}^{M_{i}}
		\e^{\jj\frac{2\pi}{\lambda}x[\hat{m}] (\vartheta_{k_{1}^{(i)}}-\vartheta_{q_{1}^{(i)}})}\overset{\text{(d)}}{=}0,
	\end{aligned}
\end{equation}
where equality (a) holds due to one-to-one correspondence of $m=\mathbf{u}_{m}^{\mathrm{T}}\mathbf{m}+1 = \hat{m} + u_{i} M_{i}+u_{i+1} M_{i+1}+\cdots+u_{I_{M}} M_{I_{M}}$ for $1 \leq m \leq M$ and $1 \leq \hat{m} \leq M_{i}$; equality (b) holds because according to the structure of the optimal APVs in \eqref{eq_APV}, $x[\hat{m}+u_{i} M_{i}] = x[\hat{m}] + u_{i} [\mathbf{d}_{x}]_{i}$ always holds for $1 \leq \hat{m} \leq M_{i}$, $0 \leq u_{i} \leq m_{i}-1$, and $1 \leq i \leq I_{M}$; equality (c) holds by expanding the exponential term $\left(x[\hat{m}]+[\mathbf{d}_{x}]_{i}+[\mathbf{d}_{x}]_{i+1}+\cdots+[\mathbf{d}_{x}]_{I_{M}}\right)$; and equality (d) holds because of \eqref{eq_SVO_i}.
Similarly, we can also prove that $\mathbf{a}_{k}^{\mathrm{ver}}(\mathbf{y}^{\star})^{\mathrm{H}}\mathbf{a}_{q}^{\mathrm{ver}}(\mathbf{y}^{\star})=0$ always holds for any $(k,q) \in \mathcal{P}_{2}$. In other words, for any $(k,q) \in \mathcal{P}=\mathcal{P}_{1} \cup \mathcal{P}_{2}$, either $\mathbf{a}_{k}^{\mathrm{hor}}(\mathbf{x}^{\star})^{\mathrm{H}}\mathbf{a}_{q}^{\mathrm{hor}}(\mathbf{x}^{\star})=0$ or $\mathbf{a}_{k}^{\mathrm{ver}}(\mathbf{y}^{\star})^{\mathrm{H}}\mathbf{a}_{q}^{\mathrm{ver}}(\mathbf{y}^{\star})=0$ holds. According to \eqref{eq_CVO}, we thus have $\hat{\mathbf{h}}_{k}(\mathbf{x}^{\star}, \mathbf{y}^{\star})^{\mathrm{H}}\hat{\mathbf{h}}_{q}(\mathbf{x}^{\star}, \mathbf{y}^{\star})=0$, $\forall (k,q) \in \mathcal{P}=\mathcal{P}_{1} \cup \mathcal{P}_{2}$. 

In addition, we also need to verify the inter-antenna spacing constraints \eqref{eq_problem_ist_e} and \eqref{eq_problem_ist_g}. According to \eqref{eq_APV}, we have
\begin{equation}\label{eq_inter_spacing}
	\begin{aligned}
		&x[m]-x[m-1]=(\mathbf{u}_{m} - \mathbf{u}_{m-1})^{\mathrm{T}}\mathbf{d}_{x},~2 \leq m \leq M.
	\end{aligned}
\end{equation}
There are only two possible cases of $\mathbf{u}_{m} - \mathbf{u}_{m-1}$: it is equal to $[1,0,\dots,0]$ or $[1-m_{1},\dots,1-m_{i},1,0,\dots,0]$, $1 \leq i \leq I_{M}$. For the former case, it is easy to verify $(\mathbf{u}_{m} - \mathbf{u}_{m-1})^{\mathrm{T}}\mathbf{d}_{x} = [\mathbf{d}_{x}]_{1} \geq d_{\min,x}$. For the latter case, it is easy to verify $(\mathbf{u}_{m} - \mathbf{u}_{m-1})^{\mathrm{T}}\mathbf{d}_{x} = [\mathbf{d}_{x}]_{i}-\sum_{j=1}^{i-1}(n_{j}-1)[\mathbf{d}_{x}]_{j}+d_{\min,x} \geq d_{\min,x}$. For both cases, the inter-antenna spacing over the horizontal direction is always no less that $d_{\min,x}$. Similarly, we can also prove that the inter-antenna spacing over the vertical direction is always no less that $d_{\min,y}$.	
This thus completes the proof.

\bibliographystyle{IEEEtran} 
\bibliography{IEEEabrv,ref_zhu}

\end{document}